\documentclass[twocolumn,10pt]{IEEEtran}

\IEEEoverridecommandlockouts                              % This command is only needed if 
\usepackage{algorithm}
\usepackage{algpseudocode}
\usepackage{amsmath} % assumes amsmath package installed
\usepackage{amssymb}  % assumes amsmath package installed
\usepackage[normalem]{ulem} %package for strikeout

\usepackage{bm}
\usepackage{xcolor}
\usepackage{comment}

\usepackage[font=scriptsize]{caption}
\usepackage{graphicx} % for pdf, bitmapped graphics files
\usepackage{svg}
\usepackage{import}
\usepackage[font=scriptsize]{subfig}
\usepackage{siunitx}
\newcommand{\RomanNumeralCaps}[1] 
    {\MakeUppercase{\romannumeral #1}} %roman numerals
\newtheorem{theorem}{Theorem}
\newtheorem{definition}{Definition}
\newtheorem{proposition}{Proposition}
\newtheorem{assumption}{Assumption}
\newtheorem{lemma}{Lemma}
\newtheorem{corollary}{Corollary}
\newtheorem{remark}{Remark}

\newenvironment{proof}
     {\textit{Proof:}}{\hfill $\blacksquare$\\}

\makeatletter
\let\c@author\relax
\makeatother
\usepackage[maxbibnames=10,giveninits=true,sorting=none]{biblatex} %Imports biblatex package
\AtEveryBibitem{%
  \clearfield{issn} % Remove issn
  \clearfield{doi} % Remove doi
  \clearfield{url}
  \clearfield{editor}
}
\DeclareFieldFormat*{year}{{}}
\renewbibmacro{in:}{} %removes In: in bibliography
\newcommand{\Exp}[2]{\ensuremath{\mathop{\mathbb{E}}_{\substack{{#1}}} \left[{#2}\right]}} % Expectation
\newcommand{\tr}[1]{{#1}^{\ensuremath{\mathsf{T}}}} % transpose 
\newcommand{\inv}[1]{{#1}^{\ensuremath{\mathsf{-1}}}} % inverse

\usepackage{xcolor}
\newcommand{\sxx}[1]{{\color{red}#1\ }}  % Suman's comments
\newcommand{\todo}[1]{{\color{orange}#1\ }}  % TODO
\newcommand{\nxx}[1]{{#1\ }} % black

\title{On the Feedback Law in Stochastic Optimal Nonlinear Control}
\author{Mohamed Naveed Gul Mohamed, Suman Chakravorty, Raman Goyal, and Ran Wang% <-this % stops a space
\thanks{The authors are with the Department of Aerospace Engineering, Texas A\&M University, College Station, TX 77843 USA. \{\tt naveed, schakrav, ramaniitrgoyal92, rwang0417\}@tamu.edu
}}

\begin{document}
\maketitle
%\begin{frontmatter}

% \author[tamu]{Mohamed Naveed Gul Mohamed}\ead{naveed@tamu.edu},  
% \author[parc]{Raman Goyal}\ead{ramaniitr.goyal92@gmail.com},
% \author[tamu]{Suman Chakravorty\thanksref{cauthor}}\ead{schakrav@tamu.edu},
% \author[rockwell]{Ran Wang}\ead{rwang0417@tamu.edu}

% \address[tamu]{Texas A\&M University, College Station, Texas, USA} 
% \address[parc]{Palo Alto Research Center, part of SRI International, California, USA}
% \address[rockwell]{Rockwell Automation, Texas, USA}

% \thanks[cauthor]{Corresponding author: Suman Chakravorty, Tel. +979-458-0064. }

\begin{abstract}
We consider the problem of nonlinear stochastic optimal control. This problem is thought to be fundamentally intractable owing to Bellman's ``curse of dimensionality". 
We present a result that shows that repeatedly solving an open-loop deterministic problem from the current state with progressively shorter horizons, similar to Model Predictive Control (MPC), results in a feedback policy that is $O(\epsilon^4)$ near to the true global stochastic optimal policy, where $\epsilon$ is a perturbation parameter modulating the noise. We also show that the optimal deterministic feedback problem has a perturbation structure such that higher-order terms of the feedback law do not affect lower-order terms and that this structure is lost in the optimal stochastic feedback problem. Consequently, solving the Stochastic Dynamic Programming problem is highly susceptible to noise, even in low dimensional problems, and in practice, the MPC-type feedback law offers superior performance even for high noise levels. 
\end{abstract}
%\end{frontmatter}
\section{\uppercase{Introduction}}
In this paper, we consider the problem of finite-time nonlinear stochastic optimal control, specifically the stochastic dynamical system:
\begin{equation}
    dx = (f(x) + g(x)u) dt + \epsilon dw, \nonumber
\end{equation}
where $w$ is a Wiener process, \nxx{$\epsilon$ is a small parameter modulating the noise,}  and the cost to be optimized is $J^{\pi}(t,x) = \Exp{}{\int_t^T c(x_t,\pi_t(x_t))dt + c_T(x_T)}$, where the incremental cost has the form $c(x,u) = l(x) + \frac{1}{2} \tr{u}Ru$, \nxx{$c_T(x_T)$ is the terminal cost,} $\pi_t(x_t)$ is a control policy and the cost is minimized over all possible such policies given the initial condition $x_0 = x$ at $t=0$. 

A large majority of sequential decision making problems under uncertainty can be posed as a nonlinear stochastic optimal control problem that requires the solution of an associated Dynamic Programming (DP) problem (in discrete time) or the Hamilton-Jacobi-Bellman (HJB) equation in continuous time, however, as the state dimension increases, the computational complexity
grows exponentially in the state dimension \cite{bertsekas1}:  the manifestation of the so-called Bellman's ``curse of dimensionality (CoD)" \cite{bellman}. Approximate DP (ADP), or alternatively, in Reinforcement Learning (RL), simulations/episodes of the system under a policy, is used to get an approximation of the cost-to-go function by sampling the domain \cite{parr3,bertsekas1}.
But, as the dimension $d$ increases, the number of samples required for evaluation grows exponentially. 
There has been recent success using the Deep RL paradigm where deep neural networks are used as nonlinear function approximators to keep the parametrization tractable \cite{RLHD1, haarnoja2018soft, fujimoto2018addressing, RLHD4, RLHD5}, however, the training times required for these approaches, and the variance of the solutions, is still prohibitive.
Hence, the primary problem with ADP/ RL techniques is the CoD inherent in the complex representation of the cost-to-go function, and the exponentially large number of evaluations required for its estimation resulting in high solution variance which makes them unreliable and inaccurate. 

In the case of continuous state, control, and observation space problems, the
Model Predictive Control \cite{Mayne_1, Mayne_2} approach has been used widely in the control system and robotics community. For deterministic
systems, the process results in solving the original DP problem in a recursive
online fashion. However, stochastic control problems, and the control of
uncertain systems in general, is still an unresolved problem in MPC. As
succinctly noted in \cite{Mayne_1}, the problem arises due to
the fact that in stochastic control problems, the MPC optimization at every
time step cannot be over deterministic control sequences, but rather has to be
over feedback policies, which is, in general, difficult to
accomplish since a tractable parameterization of such policies to perform the optimization over, is, in general, unavailable. Thus, the tube-based MPC approach, and its stochastic counterparts,
typically consider linear systems \cite{T-MPC1, T-MPC2,T-MPC3} for which a
linear parametrization of the feedback policy suffices but the methods become intractable when dealing with nonlinear systems \cite{Mayne_3}. 
In more recent work, event-triggered MPC \cite{ETMPC1, ETMPC2} keeps the online planning computationally efficient by triggering replanning in an event driven fashion rather than at every time step.
We note that event-triggered MPC inherits the same issues mentioned above with respect to the stochastic control problem, and consequently, the techniques are intractable for nonlinear systems. %\nxx{There has been recent work showing the near-optimality of MPC with a perturbation analysis \cite{bounded_regret_mpc, bounded_regret_ltv, superconvergence_mpc}, but this work considers a deterministic problem setting with unknown model parameters in the system dynamics, and the regret bound provided is with respect to the controller that has perfect knowledge of the model, in contrast, we show the near-optimality of the deterministic feedback to the optimal stochastic law.}

The fundamental problem is that, albeit solving the open-loop problem via the Minimum Principle (MP) is much easier, solving for the optimal feedback control under uncertainty requires the solution of the DP equation, which is intractable. \\
%Moreover, this also begs the question, since all systems are subject to uncertainty, what is the utility of deterministic optimal control?\\ %We answer this fundamental conundrum of feedback control in this paper. Inherently, feedback control is representation dependent and constantly shifting the frame of reference of the control to the current state unifies the MP with DP.
\textbf{Contributions:} In this work, we seek to resolve this basic conundrum. We establish that the performance of the optimal deterministic feedback law found by neglecting the process noise, when applied to the stochastic system, is near-optimal, to $O(\epsilon^4)$ in the small noise parameter $\epsilon$, to the performance of the true optimal policy for the system. 
Utilizing the method of characteristics to analyze the HJB PDE, we show that the deterministic nominal/ open-loop solution obtained by satisfying the Minimum Principle \cite{bryson} is globally optimal given that the HJB equation admits a smooth solution. Further, we do a perturbation expansion of the feedback policy to obtain the equations governing linear and higher-order perturbation feedback terms of the optimal deterministic and stochastic policy. These equations show that the deterministic feedback law has a perturbation structure in that higher-order feedback terms do not affect lower-order terms, which is lost in the stochastic problem. Next, we establish that the performance of the MPC-type approach of solving the open loop problem (over progressively shorter horizons) results in $O(\epsilon^4)$ near-optimality to the performance of the true optimal policy for the stochastic system. We also establish that the performance of the optimal deterministic linear perturbation feedback is also near-optimal to $O(\epsilon^4)$.
%This linear feedback law is also shown to be near-optimal to $O(\epsilon^4)$ for a nonlinear stochastic system. 
Finally, albeit the MPC law is only ``near-optimum" theoretically, our computational experiments show that the MPC law has superior performance than the stochastic law, obtained by solving the stochastic HJB problem computationally, showing the computational intractability of the stochastic HJB problem owing to the loss of the perturbation structure. To the best of our knowledge, all of the above results are novel.

In contrast to our prior work \cite{parunandi2019TPFC}, we show fourth order near-optimality of the linear perturbation feedback, the perturbation structure of the deterministic and stochastic feedback law, and analytical as well as empirical evidence regarding the superiority of MPC to stochastic DP. The current manuscript expands on our previously published conference paper \cite{mohamed2022acc}, giving detailed proofs of all our developments and provides a comprehensive empirical (computational) evaluation of the theoretical developments. Finally, we note that there is very recent work from the MPC literature that utilizes a shrinking horizon MPC approach \cite{Diehl-e4} and proves a similar fourth-order near optimality result; however, we showed the result earlier \cite{mohamed2022acc, mohamed2020optimality} and our method of proof is different. Further, the other contributions listed above are unique to this paper.
%Finally, allying this with the MPC logic of replanning at every time step, allows us to establish the global optimality of the MPC feedback law. Perhaps most importantly, this allows us to unify, at least under the assumptions of this paper, the Minimum Principle and Dynamic Programming in the context of Stochastic Control, i.e., if we re-plan constantly, the Minimum Principle is equivalent to Dynamic Programming. 

The rest of the document is organized as follows: Section~\ref{sec:2} formulates the problem and provides some background on the HJB. We give a heuristic overview of the results established in this paper in Section \ref{sec:heuristic}, we recommend that the reader peruse this section carefully for easier comprehension of the details in the rest of the paper. Section~\ref{sec:3} presents the perturbation analysis of the optimal feedback control problem. In Section~\ref{sec:4}, we establish that the MPC and linear perturbation feedback law are near-optimal. We illustrate our results numerically in Section~\ref{section:results} using simple 1-dimensional examples as well as more practical examples from nonlinear robotic planning. 

\section{\uppercase{Preliminaries}}
\label{sec:2}
The following outlines the finite time stochastic optimal control problem formulation, and the associated deterministic problem, along with the associated Dynamic Programming (DP) problems that we shall study in this work.

\subsection{Problem Formulation}
For a dynamical system, we denote the state and control vectors by $x \in \  \mathbb{R}^{n_x}$ and $u \in \  \mathbb{R}^{n_u}$ respectively. The dynamics of the system is governed by the stochastic differential equation (SDE):
\nxx{
\begin{equation}
   dx= (f(x)+ g(x)u)dt + \epsilon dw,
    \label{eq:model}
\end{equation}
where $w \in \mathbb{R}^{n_x}$ is a Wiener process with covariance $Q \in \mathbb{R}^{n_x \times n_x}$, and $\epsilon$ is a small parameter modulating the noise amplitude to the system and affects the signal-to-noise ratio.} 

%\textit{1-Dimensional/ Scalar case:} For the sake of simplicity in presenting the results, we will consider the scalar or 1-dimensional version of the problem, i.e., $n_x = n_u = 1$. The final results for the vector case will also be provided.

%\subsection{Stochastic optimal control problem} %
The stochastic optimal control problem for an initial state $x_0$ at time $t=0$ is defined as:
\begin{equation} \label{SOC}
    J(0,x_0) = \min_{\Pi} \ \Exp{}{\int^{T}_{0} c(x_t, \pi_t (x_t))dt + c_T(x_T)},
\end{equation}
subject to the SDE \eqref{eq:model},
where the optimization is over a family of  time-varying feedback policies $\Pi := \{ \pi_t(x);t\in[0,T]$\};
$J(\cdot, \cdot): \mathbb{R} \times \mathbb{R}^{n_x} \rightarrow \mathbb{R}$  is the cost function on applying the optimal policy $\pi^{*}$; $c(\cdot,\cdot): \mathbb{R}^{n_x} \times \mathbb{R}^{n_u} \rightarrow \mathbb{R} $  is the incremental cost function; and $c_T(\cdot): \mathbb{R}^{n_x} \rightarrow \mathbb{R}$ is the terminal cost function; where $T$ is the ``finite time horizon" of the problem.

%\subsection{Assumptions}
 We shall make the following assumptions in the rest of the paper, and unless otherwise stated, all results assume the following.
\begin{assumption}{(A1)} \label{assump:1}
\textit{Cost Structure.} We assume that the incremental cost $c(x,u)$ is quadratic in the control variable, i.e., $c(x,u) = l(x) + \frac{1}{2}\tr{u}Ru$, with $R$ positive definite. The matrix $R$ will be replaced by $r$ for the scalar case.
\end{assumption}
\begin{assumption}{(A2)}\label{assump:2}
\textit{Smoothness.} We shall also assume that all the involved functions: $f(x), g(x), l(x), c_T(x), \pi_t(x)$ are five times continuously differentiable ($\mathcal{C}^5$) in their arguments.
\end{assumption}
%
%\subsection{Dynamic Programming}
%
\nxx{The continuous time DP or the stochastic Hamilton-Jacobi-Bellman (HJB) equation for the system in \eqref{eq:model} is given by \cite{OPT_todorov}
\begin{equation}\label{DP_C}
    -\frac{\partial J}{\partial t} = \min_u H(x,u) + \frac{\epsilon^2}{2} \sum_i\sum_j \frac{\partial^2 J}{\partial x_i \partial x_j}Q_{ij},
\end{equation}
where, $J= J(t,x)$ is the cost-to-go function, $J(T,x) = c_T(x)$ is the terminal condition, $H(x,u) = l(x) + \frac{1}{2}\tr{u}Ru + \tr{\frac{\partial J}{\partial x}} (f(x)+g(x)u)$ is the Hamiltonian of the system, and $Q= [Q_{ij}]$ is the intensity of the vector Wiener process.\\
Let $u(t,x)$ denote the corresponding optimal policy.
Then, it is sufficient that the optimal control $u$ satisfies the first-order necessary condition (since the Hamiltonian $H(x,u)$ is strictly quadratic in $u$):
\begin{equation}\label{SOP}
u= -R^{-1} \tr{g(x)} J^x, ~\text{where} ~ J^x = \frac{\partial J(t,x)}{\partial x}.
\end{equation}
The deterministic problem is a special case of the stochastic problem, where $\epsilon =0$ and the same cost as in \eqref{SOC}, except there is no expectation due to the lack of stochasticity. Utilizing essentially identical arguments as for the stochastic case, the optimal cost-to-go of the deterministic system, $\phi(t,x)$, satisfies the deterministic HJB equation:
\begin{equation}\label{detDP}
   - \frac{\partial{\phi}}{\partial t}= \min_{u}H(x,u),
\end{equation}
where the terminal condition $\phi(T,x) = c_T(x)$, and the optimal control is given by: 
\begin{align}
u^d = -R^{-1} \tr{g(x)}\phi^x, \text{where} ~\phi ^x = \frac{\partial \phi}{\partial x}. \label{DOP}
\end{align}}
The HJB equation above in \eqref{DP_C}/\eqref{detDP}, in some cases, might not have a smooth solution. In this paper, we shall consider only the case where it has a smooth solution and make the following assumption for our subsequent work. We note that this is a standard assumption in many references \cite{Lewis1,Lewis2,nsfadp}.
\begin{assumption}{(A3)}\label{assump:3}
    There exists a twice differentiable ($\mathcal{C}^2$) cost-to-go function that satisfies the stochastic HJB equation \eqref{DP_C} and a continuously differentiable function ($\mathcal{C}^1$) that satisfies the deterministic HJB \eqref{detDP}. 
\end{assumption}

\subsection{Heuristic Overview}\label{sec:heuristic}
In the following, we provide a heuristic overview of the main results of this paper. We use a scalar state in the following for simplicity.

Given sufficient smoothness of the involved functions, and an initial state $x_0$, any feedback policy has the perturbation expansion : $\pi_t(x_t) = \bar{u}_t + K_t^1 \delta x_t + K_t^2 \delta x_t^2 + \cdots,$ where $\bar{u}_t$ is the nominal action, i.e., control action under zero process noise and initial condition uncertainty, and the state perturbation $\delta x_t = x_t - \bar{x}_t$, where $x_t$ is the state under process and initial condition uncertainty while $\bar{x}_t$ is the state evolution under the nominal action $\bar{u}_t$.\\
In \textit{Section~\ref{sec:III-A}}, given the parameter $\epsilon$ modulating the process noise, we show that the closed loop cost of a feedback policy can be expanded as a series in $\epsilon^2$ (owing to the Wiener process noise): $J^{\pi}(t,x) = J^{\pi,0} (t,x) + \epsilon^2 J^{\pi,1} (t,x) + \epsilon^4 J^{\pi,2} (t,x) + \cdots.$ Further, we show that all terms up to $O(\epsilon^2)$ are solely due to the nominal ($\bar{u}_t$) and linear part ($K_t^1 \delta x_t$) of the perturbation expansion of the policy $\pi_t(\cdot)$.\\
In \textit{Section~\ref{Oe4}}, owing to the above fact, the costs of both the optimal deterministic policy and the optimal stochastic policy, when applied to the stochastic system, can be respectively expanded as: $\varphi(t,x) = \varphi^0(t,x) + \epsilon^2 \varphi^1 (t,x) + \epsilon^4 \varphi^2(t,x) + \cdots,$ and $J(t,x) = J^0(t,x) + \epsilon^2 J^1(t,x) + \epsilon^4  J^2 (t,x) + \cdots.$ Then, we show that $\varphi^0(t,x) = J^0(t,x)$, and $\varphi^1(t,x) = J^1(t,x)$, since they satisfy the same PDEs with the same terminal conditions. Thus, the performance of the deterministic and stochastic policy agree to the fourth order in the small parameter $\epsilon$, i.e., the performance of the deterministic policy is near-optimal for low noise levels. Further, owing to the interpretation of the terms $\varphi^0(t,x)$ and $\varphi^1(t,x)$ from Sec.~\ref{Oe4}, it follows that the performance of the deterministic linear perturbation feedback given by $\pi_t^{d,l}(x_t) = \bar{u}_t + K_t^{1} \delta x_t$ is also near optimal to fourth order in $\epsilon$ to the stochastic policy.\\
However, solving the deterministic HJB is intractable, and thus, in \textit{Section~\ref{sec:MOC}}, for the optimal deterministic policy, we seek to find the nominal optimal actions $\bar{u}_t$, the optimal linear feedback term $K_t^1$, and the optimal higher order terms $K_t^2, K_t^3 \cdots$ such that we can get a local approximation of the policy. We show that satisfying the Minimum principle is also sufficient for determining $\bar{u}_t$ if the HJB solution is smooth. Further, we determine the equations satisfied by the linear and higher order feedback terms. These equations show a perturbation structure in the sense that the calculation of $K_t^1$ is only affected by $\bar{u}_t$ and not by the higher order terms $K_t^2, K_t^3 \cdots$; $K_t^2$ is only affected by $\bar{u}_t$ and $K_t^1$, and so on. This implies that the deterministic feedback policy can be calculated exactly to any order $n$ without regard to terms beyond $n$.\\
In \textit{Section~\ref{sec.3D}}, we perform the same exercise for the stochastic policy and show that unfortunately the perturbation structure breaks down: the calculation of $\bar{u}_t$ depends on $K_t^1, K_t^2$, the calculation for $K_t^1$ depends on $\bar{u}_t,K_t^2, K_t^3,$ and so on. Thus, the optimal stochastic policy has to be calculated to a high enough order for accuracy, and unlike the deterministic case, there is no notion of an accurate local (say linear) approximation of said policy.

In Section~\ref{sec:4}, we use the results from Section~\ref{sec:3} to show that a shrinking horizon MPC (MPC-SH) approach, which is essentially a recursive online way of applying the optimal deterministic feedback policy to the stochastic system, is near optimal to fourth order in the parameter $\epsilon$ to the optimal stochastic policy. By the same token, the linear deterministic perturbation feedback policy is also near optimal to fourth order in $\epsilon$. This result shows that the MPC-SH and linear perturbation feedback are good approximations of the optimal stochastic policy for low noise levels.

In Section \ref{section:results}, we present extensive empirical (computational) results supporting the theoretical developments in Sections~\ref{sec:3} and \ref{sec:4}. In particular, we consider very simple one dimensional examples, where we can presumably perform very accurate calculations for solving the stochastic HJB using finite differences (FD), we show that the MPC and stochastic policy agree for low values of $\epsilon$. However, at higher noise levels, albeit the stochastic policy is supposed to perform better than MPC, in practice, the performance of the computed policy is inferior to MPC owing to the fundamental computational intractability of the stochastic problem as theorized in Section~\ref{sec:3}.

\section{\uppercase{A Perturbation Analysis of Optimal Feedback Control}}\label{sec:3}
In the following four subsections, we establish four basic results that show the structure of the optimal feedback law for nonlinear deterministic and stochastic systems and their computational implications. %We shall use in section 4 to establish the near optimality of the MPC law in Section~\ref{sec:4}. In the following, we give a heuristic overview of the results.

\subsection{Characterizing the Performance of a Feedback Policy}\label{sec:III-A}
\nxx{
In order to derive the results in this section, we first discretize the SDE in Eq.~\eqref{eq:model} via a Forward Euler approximation \cite[Ch.9]{kloedon_numerical_sde} with discretization time $\Delta t$:
\begin{align} 
\label{eq.0.1}
%\label{eq:dynamics1}
x_{k+1} &= x_k + (f(x_k) + g(x_k) u_k) \Delta t + \epsilon w_k \sqrt{\Delta t} + o(\Delta t),
    %J^{\pi^{*}}(x_0) &= \lim_{\Delta t \rightarrow 0}\min_{\Pi} \ \Exp{}{\sum^{N}_{k=0} c(x_k, \pi_k (x_k))\Delta t + c_T(x_N)},
\end{align}
where $\epsilon$ is a small perturbation parameter, $w_k$ is a white noise sequence with covariance $Q = I_{n_x \times n_x}$, $k= 0,1 \cdots N$, where $N = T/\Delta t$, and $||o(\Delta t)|| \rightarrow 0$ as $\Delta t \rightarrow 0$. At the end of this Section, we will obtain the continuous time result by letting $\Delta t \rightarrow 0$. For notational convenience, we shall not explicitly write the $o(\Delta t)$ term in the following. }

%and the sampling time $\Delta t$ is small enough that the $O(\Delta t ^\alpha)$ terms are negligible for $\alpha > 1$. The noise term above stems from Brownian motion, and hence the $\sqrt{\Delta t}$ factor. We also assume that the instantaneous cost $c(\cdot, \cdot)$ has the following simple form,  
%\begin{align}
%\label{eq:cost1}
%$
%c(x,u) = (l(x) + \frac{1}{2} u'Ru)\Delta t,
%$ where $R$ is symmetric and $R \succ 0$.
%The main reason to use the above assumptions is to simplify the Dynamic Programming (DP) equation governing the optimal cost-to-go function of the system developed in section~\ref{Oe4}.\\

Let us also consider a noiseless version of the system dynamics given by \eqref{eq.0.1}, obtained by setting $w_k = 0$ for all $k$: $\bar{x}_{k+1} = \bar{x}_k + (f(\bar{x}_k) + g(\bar{x}_k) \bar{u}_k) \Delta t$, where we denote the ``nominal'' state trajectory as $\bar{x}_{k}$ and the ``nominal'' control as $\bar{u}_{k}$, with $\bar{u}_{k} = \pi_{k\Delta t}(\bar{x}_{k})$, where $\{\pi_{k\Delta t}(\cdot), k= 0,1\cdots N\}$ is a discretization of a given continuous-time control policy $\Pi= \{\pi_t(x), t \in [0,T]\}$. In the following, to simplify notation, we shall drop the explicit reference to the discretization time $\Delta t$ while denoting the discretized policy as $\{\pi_k(x)\}, k=0,1,\cdots, N$.

Under A\ref{assump:2}, $f(\cdot), g(\cdot)$ and $\pi_{k}(\cdot)$ are sufficiently smooth, so we expand the dynamics about the nominal trajectory using a Taylor series. Denoting  $\delta x_k = x_k - \bar{x}_k, \delta u_k = u_k - \bar{u}_k$, we can express, 
\begin{align}
\delta x_{k+1} &= A_k \delta x_k + B_k \delta u_k + S_k(\delta x_k) + \epsilon w_k \sqrt{\Delta t}, \label{eq.2}\\
\delta u_{k} &=  K_k \delta x_k + \tilde{S}_k(\delta x_k), \label{eq.3}
\end{align}
where $A_k = I_{n_x \times n_x} + \frac{\partial (f(x) + g(x)u)\Delta t}{\partial x}|_{\bar{x}_k, \bar{u}_k}$, \\$B_k = \frac{\partial (f(x) + g(x)u)\Delta t}{\partial u}|_{\bar{x}_k, \bar{u}_k} = g(\bar{x}_k)\Delta t$, $K_{k} = \frac{\partial \pi_{k}}{\partial x}|_{\bar{x}_k}$, and  $S_k(\cdot), \tilde{S}_k(\cdot)$ are second and higher order terms in the respective expansions. 

\nxx{Using \eqref{eq.2} and \eqref{eq.3}, we can write the closed-loop dynamics  of the trajectory $(\delta x_{k})^{N}_{k=1}$ as, 
\begin{align}
\label{eq.6}
\delta x_{k+1} = \underbrace{(A_k+B_kK_k)}_{\bar{A}_k} \delta x_k &+ \underbrace{B_k\tilde{S}_k(\delta x_k) + S_k(\delta x_k)}_{\bar{S}_k(\delta x_k)} \nonumber\\
&+ \epsilon w_k \sqrt{\Delta t}, 
\end{align} 
where $\bar{A}_k$ represents the linear part of the closed-loop system and the term $\bar{S}_k(\cdot)$ represents the second and higher order terms in the closed-loop system.}

Similarly, we can expand the instantaneous cost $c(x_{k}, u_{k})$ about the nominal values $c(\bar{x}_{k}, \bar{u}_{k})$ as,
\begin{align}
c(x_k,u_k)\nxx{\Delta t} &= \Big({l}(\bar{x}_{k}) + l^x_k \delta x_k + H_k(\delta x_k) + \nonumber\\ 
&\frac{1}{2}\tr{\bar{u}_k}R\bar{u}_k +  \tr{\delta u_k}R\bar{u}_k + \frac{1}{2} \tr{\delta u_k}R\delta u_k \Big)\Delta t,\label{eq.4}\\
c_{T}(x_{N}) &= {c}_{T}(\bar{x}_{N}) + C_T \delta x_N + H_T(\delta x_N),\label{eq.5}
\end{align}
where $l^x_k = \frac{\partial l}{\partial x}|_{\bar{x}_k}$, $C_T  = \frac{\partial c_T}{\partial x}|_{\bar{x}_N}$, and $H_k(\cdot)$ and $H_T(\cdot)$ are second and higher order terms in the respective expansions. The closed-loop incremental cost given in  \eqref{eq.4} can be expressed as
\begin{multline*}
c(x_k,u_k) \nxx{\Delta t} = \underbrace{\{{l}(\bar{x}_{k}) + \frac{1}{2}\tr{\bar{u}_k}R\bar{u}_k\}\Delta t}_{\bar{c}_k} + \nonumber\\
\underbrace{[l^x_k + \tr{\bar{u}_k}RK_k]\Delta t}_{\bar{C}_k} \delta x_k + \bar{H}_k(\delta x_k),
\end{multline*}
where $\bar{H}_k(\delta x_k)$ are the second and higher order terms.
Therefore, the cumulative cost of any given closed-loop trajectory $(x_{k}, u_{k})^{N}_{k=0}$  can be expressed as,
$
\mathcal{J}^{\pi}(x_0) = \sum^{N}_{k=0}c(x_{k}, \pi_{k}(x_{k})) \nxx{\Delta t} + c_{T}(x_{N})
$, which can be written in the following form:
\begin{align}
\label{eq.9a}
\mathcal{J}^{\pi} (x_0) &=\sum_{k=0}^N \bar{c}_k + \sum_{k=0}^N \bar{C}_k \delta x_k + \sum_{k=0}^N \bar{H}_k(\delta x_k),
\end{align}
where $\bar{c}_{N} = c_{T}(\bar{x}_{N}),  \bar{C}_{N} = C_{T}$.

We first show the following critical result.
\textit{Note:} The proofs for the results shown here are given in the appendix. 
\begin{lemma} 
\label{L1}
Given any sample path, the state perturbation equation given in  \eqref{eq.6} can be equivalently characterized  as
\begin{align}
\label{eq:mod-pert-1}
\delta x_{k}  = \delta x_k^l + e_k, ~ \delta x_{k+1}^l = \bar{A}_k \delta x_k^l + \epsilon w_k \sqrt{\Delta t}
\end{align} 
where $e_k$ is an $O(\epsilon^2)$ function that depends on the entire noise history $\{w_0,w_1,\cdots w_k\}$ and $\delta x_k^l$ evolves according to the  linear closed-loop system. Furthermore, $e_k = e_k^{(2)} + O(\epsilon^3)$, where $e_k^{(2)} = \bar{A}_{k-1} e_{k-1}^{(2)} + \begin{bmatrix}\tr{\delta x_{k-1}^{l}}\bar{S}_{1,k-1}^{(2)} \delta x_{k-1}^l\\ \vdots \\
\tr{\delta x_{k-1}^{l}}\bar{S}_{n_x,k-1}^{(2)} \delta x_{k-1}^l\end{bmatrix}$, $e_0^{(2)} = 0$, $\bar{S}_{i,k-1}^{(2)}$ represents the Hessian matrix corresponding to the Taylor series expansion of the $i^{th}$ component of the  vector valued function $\bar{S}_{k-1}(\cdot)$. 
\end{lemma}
% \begin{proof} See arXiv report  \cite{mohamed2020optimality}.
% \end{proof}

%Using  \eqref{eq:mod-pert-1} in  \eqref{eq.9a}, we can obtain the cumulative cost of any given closed-loop trajectory as,
%\begin{align}
%J^{\pi} = \underbrace{\sum_{t=1}^T \bar{c}_t }_{\bar{J}^{\pi}} + \underbrace{\sum_{t=1}^T \bar{C}_t \delta x_t^l}_{\delta J_1^{\pi}} + %\nonumber\\
%\underbrace{\sum_{t=1}^T \bar{H}_t(\delta x_t) + \bar{C}_t \bar{\bar{S}}_t}_{\delta J_2^{\pi}}. \label{eq.9b}
%\end{align}
Next, we have the following result for the expansion of the cost-to-go function $\mathcal{J}^{\pi}(x_0)$.
\begin{lemma}
\label{L2}
Given any sample path, the cost-to-go under a policy can be expanded about the nominal as: 
\begin{align*} 
\mathcal{J}^{\pi} (x_0)  =& \underbrace{\sum_k \bar{c}_k}_{\bar{J}^{\pi}} + \underbrace{\sum_k \bar{C}_k \delta x_k^l}_{\delta J_1^{\pi}} + \underbrace{\sum_k \tr{\delta x_k^{l}}\bar{H}_k^{(2)} \delta x_k^l+ \bar{C}_k e_k^{(2)}}_{\delta J_2^{\pi}} \\ &+ O(\epsilon^3),
\end{align*}
where $\bar{H}_k^{(2)}$ denotes the second order coefficient of the Taylor expansion of $\bar{H}_k(\cdot)$.
\end{lemma}
% \begin{proof} See arXiv report  \cite{mohamed2020optimality}.
% \end{proof}

Finally, we have the following result characterizing the cost of the policy as the discretization time $\Delta t \rightarrow 0$. 

\begin{proposition}  
\label{prop1}  Under A\ref{assump:2}, and given that the closed loop system under the policy $\pi_t(\cdot)$ has a solution over the interval $[0,T]$ (A\ref{assump:3}), the mean of the cost-to-go function obeys: $\lim_{\Delta t \rightarrow 0}E[\mathcal{J}^{\pi}(x_0)] \equiv J^{\pi}(x_0) ={J}^{\pi,0}(x_0) + \epsilon^2 {J}^{\pi,1}(x_0) + \epsilon^4 {J}^{\pi,2}(x_0) + \mathcal{R}^{\pi}(x_0)$, for some constants ${J}^{\pi,k}(x_0)$, $k = 0,1,2$, where $\mathcal{R}^{\pi}(x_0)$ is $o(\epsilon^4)$, i.e., $\lim_{\epsilon \rightarrow 0}\epsilon^{-4}\mathcal{R}^{\pi}(x_0)=0$. Furthermore, the term ${J}^{\pi,0}$ arises solely from the nominal control sequence while ${J}^{\pi,1}$ is solely dependent on the nominal control and the linear part of the perturbation closed-loop.
%$%$ \begin{align*}
%\tilde{J}^{\pi} = \mathbb{E}[J^{\pi} ] =  \bar{J}^{\pi} + O(\epsilon^2), \text{and}\,
 %\text{Var}(J^{\pi}) =  \underbrace{\text{Var}(\delta J_{1}^{\pi})}_{O(\epsilon^{2})} +  O(\epsilon^4). 
 %$%\end{align*}
 \end{proposition}
%  \begin{proof} See arXiv report  \cite{mohamed2020optimality}.
%  \end{proof}
 
\begin{remark}
The interpretation of the result above is as follows: it shows that the $\epsilon^0$ term, ${J}^{\pi,0}$, in the cost, stems from the nominal action of the control policy, the $\epsilon^2$ term, $J^{\pi,1}$, stems from the linear feedback action of the closed-loop, while the higher order terms stem from the higher order terms in the feedback law. In the next section, we use the HJB equation to find the equations satisfied by these terms.
\end{remark}

\begin{remark}In the above development, we have derived the expression for the cost-to-go of a policy from the initial state $x_0$ at the initial time $t=0$, i.e., the above expressions are for $J^{\pi}(0,x_0)$, however, such an expression is also valid for any pair $(t,x)$ simply by repeating the above development starting at time $t$ from state $x$, i.e., any $J^{\pi}(t,x) = J^{\pi,0}(t,x) + \epsilon^2 J^{\pi,1}(t,x) + \epsilon^4 J^{\pi,2}(t,x) + \mathcal{R}^{\pi}(t,x)$.
\end{remark}

\subsection{A Closeness Result for Optimal Stochastic and Deterministic Control} \label{Oe4}
Recall the stochastic and deterministic HJB equations \eqref{DP_C}, \eqref{detDP} from Section~\ref{sec:2}, and the associated optimal control policies \eqref{SOP} and \eqref{DOP}. For simplicity, we consider the scalar case here, the vector case is detailed in the Appendix.
Let $\varphi(t,x)$ denote the cost-to-go of the deterministic policy when applied to the stochastic system, i.e., $u^d$ applied to Eq. \eqref{eq:model}. Note that the cost-to-go of the deterministic policy applied to the stochastic system, $\varphi(t,x)$, is different from the deterministic cost-to-go $\phi(t,x)$, and $\varphi(t,x)$ satisfies a policy evaluation equation \cite{bertsekas1}. 
%In order to derive it, first consider a time discretization $\Delta t$.  The time discretized policy satisfies the discrete time policy evaluation equation :
Similar to the stochastic HJB, the continuous time policy evaluation equation for $\varphi(t,x)$ can be written as:
% \begin{equation}\label{detDP_eps}
%    \varphi_k(x) = c(x, u_k^d(x)) + E[\varphi_{k+1}(x')],
% \end{equation}
% where $\varphi_k(x) = \varphi(k\Delta t, x)$, $u_k^d(x) = u^d(k\Delta T, x)$, and $x' = x+ (f(x) + g(x)u_k^d(x)) \Delta t + \epsilon w_t \sqrt{\Delta t}$.
% Using similar arguments as we used previously to derive the stochastic HJB, we may write for small $\Delta t$:
% \begin{equation}
%     E[\varphi_{k+1}(x')] \approx \varphi_{k+1}(x) + \varphi_{k+1}^x (f(x)+ g(x)u_k^d)\Delta t + \frac{\epsilon^2}{2} \varphi_{k+1}^{xx} \Delta t.
% \end{equation}
% Substituting the above equation into \eqref{detDP_eps}, rearranging terms, dividing by $\Delta t$ and taking the limit as $\Delta t \rightarrow 0$ gives the continuous time policy evaluation equation:
\begin{equation}\label{CPE}
    \frac{\partial \varphi}{\partial t} = l(x) + \frac{1}{2}r(u^d)^2 + \varphi^x(f(x)+g(x)u^d) + \frac{\epsilon^2}{2} \varphi^{xx},
\end{equation}
where $u^d = -\frac{1}{r}g(x) \phi^x$.
Then, we have the following key result. An analogous version of the following result was originally proved in a seminal paper \cite{fleming1971stochastic} for first passage problems. We provide a simple derivation of the result for a finite time final value problem below. 
\begin{proposition}
\label{prop_e4}
The cost function of the optimal stochastic policy, $J(t,x)$, and the cost function of the ``deterministic policy  applied to the stochastic system", $\varphi(t,x)$, satisfy: $J(t,x) = J^0 (t,x) + \epsilon^2 J^1(t,x) + \epsilon^4 J^2 (t,x)+ \cdots$, and $\varphi(t,x) = \varphi^0 (t,x) + \epsilon^2 \varphi^1 (t,x) + \epsilon^4 \varphi^2(t,x) + \cdots$. Furthermore, $J^0(t,x) = \varphi^0 (t,x)$, and $J^1(t,x) = \varphi^1(t,x)$, for all $(t,x)$.
\end{proposition}

\begin{proof}
We show a sketch here for the case of a scalar state, please refer to the appendix for the complete proof.\\
Due to Proposition \ref{prop1}, the optimal cost function satisfies:
$%\begin{equation} \label{f1}
J(t,x) = J^0(t,x) + \epsilon^2 J^1(t,x) + \epsilon^4 J^2(t,x) + \cdots$. Next, we substitute the above equation into the HJB equation \eqref{DP_C}, along with the minimizing control \eqref{SOP} to obtain a perturbation expansion of the optimal cost function as a power series in $\epsilon^2$. Equating the $O(\epsilon^0)$ and $O(\epsilon^2)$ terms on both sides results in governing equations for the $J^0$ and $J^1$ terms.
%\begin{align} %\label{f4}
%-\frac{\partial J^0}{\partial t} =& l(x) + \frac{1}{2} \frac{g(x)^2}{r}(J^{0,x})^2 + \nonumber\\
%&(f(x) + g(x)\frac{-g(x)}{r} J^{0,x}) J^{0,x},\,
%\end{align}
% with the terminal condition $J^0(T,x) = c_T(x)$,
%and
%\begin{align} %\label{f5}
%-\frac{\partial J^1}{\partial t} = (f(x)+ g(x) \frac{-g(x)}{r} J^{0,x})J^{1,x}  + \frac{1}{2} J^{0,xx},
%\end{align}
%with terminal condition $J^1(T,x) = 0$.\\
We also know that the cost function of the deterministic policy when applied to the stochastic system satisfies $\varphi(t,x) = \varphi^0(t,x) + \epsilon^2 \varphi^1(t,x)+ \cdots$. Similar to above, we substitute this expression into the policy evaluation equation \eqref{CPE}, along with the deterministic optimal control expression $u^d = - \frac{1}{r} g(x)\phi^x$, to obtain the governing equations for $\varphi^0$ and $\varphi^1$. These equations, when compared with those for $J^0$ and $J^1$, are seen to be identical with the same terminal conditions thereby proving the result.
\end{proof}

\begin{remark} \label{remark:linear_fb_nearoptimality}
\textit{$\mathcal{O}(\epsilon^4)$ Near-Optimality of Linear Perturbation Feedback.} According to Proposition \ref{prop1}, we know that the $O(\epsilon^2)$ term in the perturbation expansion above stems from the linear feedback term for any policy, and thus, the same is true for the optimal deterministic policy. Given an initial state $x_0$, let $(\bar{x}(t),\bar{u}(t))$ denote the optimal nominal trajectory under the deterministic feedback law and let $K_t$ denote the linear feedback corresponding to the expansion of the feedback law about this nominal trajectory. Therefore, it follows that if one applies the perturbation linear feedback law $u(t,x_t) = \bar{u}_t + K_t \delta x_t$, where the feedback acts on the perturbation from the nominal, $\delta x_t = x_t - \bar{x}_t$, starting at the initial state $x_0$, then the performance of this linear feedback policy is also within $O(\epsilon^4)$ of the optimal stochastic policy. 
\end{remark}
In the next section, we show how to obtain the optimal linear feedback term $K_t$, and higher order terms of the feedback law.

%The result above has used the fact that the noise sequence $w_t$ is white. However, this is not necessary to show that $J_t^0(x) = \varphi_t^{0}(x)$ for all $(t,x)$. This allows for the following general result.
%\begin{proposition}\label{prop_gen_noise}
%Let $\{w_t\}$ be a noise sequence that is dependent on the past values of the state $x_{\tau}$ and the noise $w_{\tau}$, $\tau < t$. The cost function of the optimal stochastic policy, $J_t$, and the cost function of the ``deterministic policy  applied to the stochastic system", $\varphi_t$, satisfy: $J_t(x) = J_t^0 (x) + \epsilon J_t^1(x) + \epsilon^2 J_t^2(x) + \epsilon^3 J_t^3 (x)+ \cdots$, and $\varphi_t(x) = \varphi_t^0 (x) + \epsilon^1 \varphi_t^1 (x) + \epsilon^3 \varphi_t^3(x) + \cdots$. Furthermore, we have $J_t^0(x) = \varphi_t^0 (x)$.
%\end{proposition}
%\begin{proof} See Appendix.
%\end{proof}

%\begin{remark}
%The results above show that the cost due to the nominal action, $J_t^0(x)$ (for any general noise sequence), and the cost due to the linear feedback action, $J_t^1(x)$ (for the white noise case), are the same for the optimal deterministic and optimal stochastic policies, when acting on the stochastic system, given they both start at state $x$ at time $t$. However, it is not necessarily true that the actual control action for the two policies is the same at time $t$ and state $x$. It is precisely this problem that is studied in the next Section using the classical Method of Characteristics.
%\end{remark}
%
\subsection{A Perturbation Expansion of Deterministic Optimal Feedback Control}\label{sec:MOC}
In this section, we will use the classical Method of Characteristics \cite{Courant-Hilbert} to derive results regarding the perturbation structure of deterministic optimal control problem. %In particular, we will show that satisfying the Minimum Principle is sufficient to assure us of a global optimum for the open-loop problem if the HJB admits a smooth solution. More importantly, we shall show that the deterministic cost-to-go function has a perturbation structure in that the higher-order terms do not affect the lower-order terms in a Taylor expansion of the optimal feedback law. We also obtain the equations governing the linear and higher-order feedback terms in the optimal feedback law. 
For simplicity, we derive the following for the case of a scalar state, please see the Appendix for the vector case.

Let us recall the Hamilton-Jacobi-Bellman (HJB) equation in continuous-time under the same assumptions as above, i.e., quadratic in control cost $c(x,u) = l(x) + \frac{1}{2}r u^2$, and affine in control dynamics $\dot{x} = f(x) + g(x)u$ \cite{bryson}:
\begin{equation}
    \frac{\partial J}{\partial t} + l(x) - \frac{1}{2}\frac{g(x)^2}{r} (J^x)^2 + f(x)J^x = 0,
\end{equation}
where $J = J(t,x), \; J^x = \frac{\partial J}{\partial x}$, and the equation is integrated back in time with terminal condition $J(T,x) = c_T(x)$. Define $\frac{\partial J}{\partial t} = p, \; J^x= q$, then the HJB can be written as $F(t,x,J,p,q) =0$, where $F(t,x,J,p,q) = p+l(x) -\frac{1}{2}\frac{g(x)^2}{r}q^2 + f(x)q$. One can now write the Lagrange-Charpit equations \cite{Courant-Hilbert} for the HJB as:
\begin{align}
    \dot{x} &= F_q = f(x) -\frac{g(x)^2}{r}q, \label{L-C-1}\\
    \dot{q} &= -F_x -qF_J = -l^x + \frac{g(x)g^x}{r}q^2 - f^xq,\label{L-C-2}
\end{align}
with the terminal conditions $x(T) = x_T, \; q(T) = c_T^x(x_T)$, where $F_x = \frac{\partial F}{\partial x}$, $F_q = \frac{\partial F}{\partial q}$, $g^x = \frac{\partial g}{\partial x}$, $l^x = \frac{\partial l}{\partial x}$, $f^x = \frac{\partial f}{\partial x}$ and $c_T^x = \frac{\partial c_T}{\partial x}$.\\
Given a terminal condition $x_T$, the equations above can be integrated back in time to yield a characteristic curve of the HJB PDE. First, we present an important fact regarding the Minimum Principle (MP) and the Characteristic ODEs above.

\begin{remark}\label{remark:global_optimum}
\textit{Convexity and Global Minimum.} Given that the HJB PDE has a smooth solution (Assumption \ref{assump:3}), the characteristic curves governed by the Lagrange-Charpit equations \eqref{L-C-1}, \eqref{L-C-2}, in $(x,q)$ space are unique, i.e., given any time $t$ and state $x_t$, there is a unique co-state $q_t(x_t)$ corresponding to the state. However, satisfying the MP from an initial state $x_0$ simply yields a characteristic curve of the HJB that passes through the state $x_0$.  Therefore, the open-loop optimal trajectory, found by satisfying the Minimum Principle is also the unique global minimum even though the open-loop problem is non-convex. This observation is stated in the following result, and a formal proof is given in the appendix.
\end{remark}

\begin{proposition}\label{OL_optimality}
\textit{Global Optimality of open-loop solution.} %Let the cost functions $l(\cdot)$, $c_T(\cdot)$, the drift $f(\cdot)$ and the input influence function $g(\cdot)$ be $\mathcal{C}^2$, i.e., twice continuously differentiable, and let a solution to \eqref{L-C-1}-\eqref{L-C-2} exist in $[0,T]$ for any terminal condition $(x_T,q_T)$. 
Under A\ref{assump:1}, A\ref{assump:3} an optimal trajectory that satisfies the Minimum Principle from a given initial state $x_0$, is the unique global minimum of the determinsitic optimal control problem starting at the initial state $x_0$.
\end{proposition}

Suppose now that one is given an optimal nominal trajectory $\bar{x}_t$, $t \in [0,T]$ for a given initial condition $x_0$, from solving the open-loop optimal control problem. Let the nominal terminal state be $\bar{x}_T$. We now expand the HJB solution around this nominal optimal solution to obtain the linear and higher order perturbation feedback terms assuming that the solution is sufficiently smooth. To this purpose, let $x_t = \bar{x}_t + \delta x_t$, for $t\in [0,T]$, where $\delta x_t$ is the perturbation term. Then, expanding the optimal cost function around the nominal yields: $J(t,x_t) = \bar{J}_t + G_t\delta x_t + \frac{1}{2}P_t\delta x_t^2+ \frac{1}{6}S_t\delta x_t^3+\cdots,$ where $\bar{J}_t = J(t,\bar{x}_t), G_t = \frac{\partial J}{\partial x_t}|_{\bar{x}_t}, P_t = \frac{\partial^2 J}{\partial x_t^2}|_{\bar{x}_t}$, $S_t = \frac{\partial^3 J}{\partial x_t^3}|_{\bar{x}_t}$. Then, the co-state $q = \frac{\partial J}{\partial x_t} = G_t + P_t \delta x_t + \frac{1}{2}S_t \delta x^2_t+ \cdots$. 
For simplicity, we assume that $g^x = 0$ (this is relaxed but at the expense of a rather tedious derivation detailed in the Appendix). Hence, 
\begin{align*}
    \underbrace{\frac{d}{dt} (\bar{x}_t + \delta x_t)}_{\dot{\bar{x}}_t + \dot{\delta x}_t} =& \underbrace{f(\bar{x}_t+ \delta x_t)}_{(\bar{f}_t + \bar{f}_t^x \delta x_t+ \frac{1}{2}\bar{f}_t^{xx} \delta x^2_t+O(\delta x_t^3))} \\
    &-\frac{g^2}{r}(G_t + P_t \delta x_t+ \frac{1}{2}S_t \delta x^2_t+O(\delta x_t^3)),
\end{align*}
where $\bar{f}_t = f(\bar{x}_t), \bar{f}_t^x = \frac{\partial f}{\partial x_t}|_{\bar{x}_t}$. Expanding in powers of the perturbation variable $\delta x_t$, the equation above can be written as (after noting that $\dot{\bar{x}}_t = \bar{f}_t - \frac{g^2}{r}G_t$ due to the nominal trajectory $\bar{x}_t$ satisfying the characteristic equation):
\begin{align} \label{x_perturb}
    \dot{\delta x}_t = (\bar{f}_t^x - \frac{g^2}{r} P_t)\delta x_t + \frac{1}{2} (\bar{f}_t^{xx} - \frac{g^2}{r}S_t )\delta x^2_t + O(\delta x_t^3).
\end{align}
Next, we have: $ \frac{dq}{dt} = -l^x - f^x q $
\begin{align}
    &\frac{d}{dt} (G_t+ P_t\delta x_t+ \frac{1}{2} S_t\delta x^2_t + O(\delta x^3)) %\nonumber\\
    = -(\bar{l}_t^x + \bar{l}_t^{xx}\delta x_t \nonumber\\
    &+\frac{1}{2} \bar{l}_t^{xxx}\delta x^2_t+ O(\delta x^3)) 
    - \Big(\bar{f}_t^x + \bar{f}_t^{xx}\delta x_t + \frac{1}{2} \bar{f}_t^{xxx}\delta x^2_t \nonumber\\
    &+ O(\delta x^3)\Big)(G_t + P_t \delta x_t + \frac{1}{2} S_t\delta x^2_t + O(\delta x^3)),
\end{align}
where $\bar{f}_t^{xx} = \frac{\partial^2 f}{\partial x^2}|_{\bar{x}_t},\bar{f}_t^{xxx} = \frac{\partial^3 f}{\partial x^3}|_{\bar{x}_t}, \bar{l}^x_t = \frac{\partial l}{\partial x}|_{\bar{x}_t},\bar{l}^{xx}_t = \frac{\partial^2 l}{\partial x^2}|_{\bar{x}_t}, \bar{l}^{xxx}_t = \frac{\partial^3 l}{\partial x^3}|_{\bar{x}_t}$.
Using $\frac{d}{dt}P_t \delta x_t = \dot{P}_t \delta x_t + P_t \dot{\delta x}_t$, $\frac{d}{dt}S_t \delta x^2_t = \dot{S}_t \delta x^2_t + 2 S_t \delta x \dot{\delta x}_t$, and substituting for $\dot{\delta x}_t$ from \eqref{x_perturb},  and expanding the two sides above in powers of $\delta x_t$ yields:
\begin{align*}
    &\dot{G}_t + (\dot{P}_t + P_t (\bar{f}_t^x - \frac{g^2}{r} P_t))\delta x_t + \frac{1}{2} \Big( P_t (\bar{f}_t^{xx} - \frac{g^2}{r} S_t) \\
    &+ \dot{S}_t + 2 S_t(\bar{f}_t^x - \frac{g^2}{r} P_t) \Big) \delta x^2 + O(\delta x^3) \\
    & \quad = -(\bar{l}_t^x +\bar{f}_t^x G_t) -(\bar{l}_t^{xx} + \bar{f}_t^x P_t +\bar{f}_t^{xx}G_t)\delta x_t \\
    &- \frac{1}{2} (\bar{l}_t^{xxx}  +\bar{f}_t^{xxx}G_t + 2 \bar{f}_t^{xx} P_t + \bar{f}_t^x S_t) \delta x^2 + O(\delta x^3).
\end{align*}
Equating the first three powers of $\delta x_t$ yields:
\begin{align}
    &\dot{G}_t + \bar{l}_t^x + \bar{f}_t^x G_t = 0, \\
    &\dot{P}_t + \bar{l}_t^{xx} + P_t \bar{f}_t^x + \bar{f}_t^x P_t - P_t \frac{g^2}{r} P_t + \bar{f}_{t}^{xx} G_t = 0, \\
    & \dot{S}_t + \bar{l}_t^{xxx} + P_t \bar{f}_t^{xx} + 2 \bar{f}_t^{xx} P_t  + \bar{f}_t^x S_t + 2 S_t \bar{f}_t^x - P_t \frac{g^2}{r} S_t \nonumber \\
    &- 2 S_t \frac{g^2}{r} P_t  + \bar{f}_{t}^{xxx} G_t = 0
\end{align}
% We can see that the first equation for $G_t$ is simply a restatement of the co-state equations in the Minimum Principle. The second equation is a Riccati-esque equation, except for the second order term $\bar{f}_t^{xx}G_t$. This clearly shows that the feedback design is not an LQR design. 
Using the first-order necessary condition $u(t,x_t) = - \frac{g}{r} J^x$, the optimal feedback law is given by: 
\begin{align}
    u(t,x_t) &= - \frac{g}{r} J^x  =  \underbrace{- \frac{g}{r}G_t}_{\bar{u}_t} \underbrace{- \frac{g}{r} P_t}_{K_t} \delta x_t \underbrace{- \frac{g}{2r} S_t}_{K^{(2)}_t} \delta x^2_t + O(\delta x_t^3) \\
     u(t,x_t) &= \bar{u}_t + K_t \delta x_t +  K^{(2)}_t \delta x^2_t +  \cdots \nonumber.
\end{align}
Thus, we see that the optimal feedback law has a perturbation structure in that the first-order terms $P_t$ do not affect the zeroth-order terms $G_t$, and the second and higher-order terms, $S_t$ etc., do not affect the first-order term $P_t$ and so on.\\
Now, we provide the final result for the general vector case with a state-dependent control influence matrix (please see the Appendix for details). %We ignore the $O(\delta x_t^2)$ and higher-order terms in the feedback law purely for notational convenience. 
\begin{definition}\label{def:system_vec}
Let the control influence matrix be given as:
$g(x) = \begin{bmatrix} g_1^1(x) \cdots g_1^p(x)\\ \ddots\\ g_n^1(x) \cdots g_n^p(x) \end{bmatrix} = \begin{bmatrix} \Gamma^1 (x) \cdots \Gamma^p (x) \end{bmatrix}$, 
i.e., $\Gamma^j$ represents the control influence vector corresponding to the $j^{th}$ input. Let $\bar{g}_t = g(\bar{x}_t)$, where $\{\bar{x}_t\}$ represents the optimal nominal trajectory. 
Further, let $f(x) = \begin{bmatrix} f_1(x) \cdots f_n(x)\end{bmatrix}^\intercal$ 
denote the drift of the system. Let $G_t = [G_t^1\cdots G_t^n]^{\intercal}$, and $ R^{-1} \bar{g}_t^{\intercal} G_t = - [\bar{u}_t^1 \cdots \bar{u}_t^p]^{\intercal}$, denote the optimal nominal co-state and control vectors respectively.
Let the Jacobian and Hessian of our system matrices be defined as:

\begin{align}
    \bar{f}_t^x &= \begin{bmatrix} \frac{\partial f_1}{\partial x_1} \cdots \frac{\partial f_1}{\partial x_n}\\ \ddots\\ \frac{\partial f_n}{\partial x_1} \cdots \frac{\partial f_n}{\partial x_n} \end{bmatrix}|_{\bar{x}_t}, ~~
    \bar{f}_t^{xx,i} = \begin{bmatrix} \frac{\partial^2 f_1}{\partial x_1 \partial x_i} \cdots \frac{\partial^2 f_1}{\partial x_n\partial x_i}\\ \ddots\\ \frac{\partial^2 f_n}{\partial x_1 \partial x_i} \cdots \frac{\partial^2 f_n}{\partial x_n \partial x_i} \end{bmatrix}|_{\bar{x}_t}, \nonumber\\
    \bar{g}_t^{x,i} &= \begin{bmatrix} \frac{\partial g_1^1}{\partial x_i} \cdots \frac{\partial g_1^p}{\partial x_i}\\ \ddots\\ \frac{\partial g_n^1}{\partial x_i} \cdots \frac{\partial g_n^p}{\partial x_i}\end{bmatrix}|_{\bar{x}_t} \label{eq:sys_matrices_vec}.
\end{align}
Similarly $\bar{\Gamma}^{j,x}_t = \nabla_x \Gamma^j|_{\bar{x}_t}$, $\bar{\Gamma}^{j,xx,i}_t= \nabla_{xx} \Gamma^j|_{\bar{x}_t}$ for the vector function $\Gamma^j$. Finally, define $\mathcal{A}_t = \bar{f}_t^x + \sum_{j=1}^p \bar{\Gamma}_t^{j,x} \bar{u}_t^j$, $\bar{l}_t^x = \nabla_x l|_{\bar{x}_t}$, and $\bar{l}_t^{xx} = \nabla^2_{xx} l |_{\bar{x}_t}$. 
\end{definition}

\begin{proposition} \label{T-PFC}
\nxx{Under A\ref{assump:1},}and given the above definitions, the following result holds for the evolution of the co-state/ gradient vector $G_t$, and the Hessian matrix $P_t$, of the optimal cost function ${J}_t(x_t)$, evaluated on the optimal nominal trajectory $\bar{x}_t,t\in[0,T]$:
\begin{align}
    \dot{G}_t &+ \bar{l}_t^x + \mathcal{A}_t^{\intercal} G_t = 0,\label{T-PFC-G} \\
    \dot{P}_t &+ \mathcal{A}_t^{\intercal} P_t + P_t \mathcal{A}_t + \bar{l}_t^{xx} \nonumber\\
    + \sum_{i=1}^n  &[\bar{f}_t^{xx,i} + \sum_{j=1}^p \bar{\Gamma}_t^{j,xx,i} \bar{u}_t^j]G_t^i
     - K_t^{\intercal} R K_t= 0, \label{T-PFC-P} \\
     K_t &= - R^{-1} [\sum_{i=1}^n \bar{g}_t^{x,i,\intercal}  G_t^i + \bar{g}_t^{\intercal} P_t], \label{T-PFC-K}
\end{align}
with terminal conditions $G_T = \nabla_x c_T|_{\bar{x}_T}$, and $P_T = \nabla_{xx}^2 c_T|_{\bar{x}_T}$ and the control input with the optimal linear feedback is given by $u_t = \bar{u}_t + K_t \delta x_t$.\\
\end{proposition}
% \begin{proof}
% See arXiv report  \cite{mohamed2020optimality}.
% \end{proof}

\begin{remark} \label{T-PFC-remarks} 
\textit{\nxx{Not standard LQR.}}
The co-state equation \eqref{T-PFC-G} above is identical to the co-state equation in the Minimum Principle \cite{bryson, Pontryagin}. However, the Hessian $P_t$ equation \eqref{T-PFC-P} is Riccati-like with some important differences: note the extra second order terms due to $\bar{f}_t^{xx,i}$ and $\bar{\Gamma}_t^{xx,i}$ in the second line stemming from the nonlinear drift and input influence vectors and an extra term in the gain equation \eqref{T-PFC-K} coming from the state dependent influence matrix. These terms are not present in the LQR Riccati equation, and thus, it is clear that this cannot be a traditional perturbation feedback design \cite[Ch. 6]{bryson}. 
\end{remark}
%\begin{remark} \label{discrete-time}\textbf{Discrete-time case.}
%For the discrete-time case with small  discretization time $\Delta t$, one would discretize the noiseless model with a forward Euler approximation as $x_{t+1} = x_t + f(x_t)\Delta t + g(x_t)u_t \Delta t $ and the above equations  as:
%\begin{align}\label{dT-PFC-G}
 %   G_t &= \bar{L}_t^x + A_t^{\intercal} G_{t+1},\\
 %   P_t &= A_t^{\intercal} P_{t+1} A_t + \bar{L}_t^{xx} + \sum_{i=1}^n  [\bar{f}d_t^{xx,i} + \\ \nonumber
 %   &\sum_{j=1}^p \bar{\Gamma}d_t^{j,xx,i} \bar{u}_t^j]G_{t+1}^i 
 %    - K_t^{\intercal} (R_t + B_t^\intercal P_{t+1} B_t) K_t,\\
 %    K_t &= - (R_t + B_t^\intercal P_{t+1} B_t)^{-1} [\sum_{i=1}^n \bar{g}d_t^{x,i,\intercal}  G_{t+1}^i + \nonumber\\
 %    & B_t^{\intercal} P_{t+1} A_t]. \label{dT-PFC-K}
%\end{align}
%where, $A_t = I + (\bar{f}_t^x + \sum_{j=1}^p \bar{\Gamma}_t^{j,x} \bar{u}_t^j) \Delta t$, $B_t = \bar{g}\Delta t$, $\bar{f}d_t^{xx,i} = \bar{f}_t^{xx,i} \Delta t$, $\bar{\Gamma}d_t^{j,xx,i} = \bar{\Gamma}_t^{j,xx,i}\Delta t$, $\bar{g}d_t^{x,i} = \bar{g}_t^{x,i} \Delta t$. 
%\end{remark}

%
\subsection{Loss of Perturbation Structure in Stochastic Control} ~\label{sec.3D}
Finally, we outline the loss of the perturbation structure in the stochastic problem. For the sake of simplicity, we only consider the scalar case, however, even this simplest case brings out the difficulty associated with stochastic control while the generalization to the vector case is relatively straightforward, albeit tedious.

Recall the stochastic HJB:
\begin{equation}\label{eq.stochastic_hjb}
    -\frac{\partial J}{\partial t} = \min_{u} [H(x,u)] + \nxx{\frac{\epsilon^2}{2}} \frac{\partial^2 J}{\partial x^2},
\end{equation}
where $H(x,u) = l(x) + \frac{1}{2}ru^2 + (f(x)+g u) \frac{\partial J}{\partial x}$ is the Hamiltonian of the system, and the equation is integrated backwards from a terminal condition $J(T,x) = c_T(x)$. For simplicity, we assume that $g$ is not state dependent in the following derivation and we also assume the noise variance $Q = 1$, which otherwise would appear in the diffusion term in \eqref{eq.stochastic_hjb}. Suppose now that we are given the optimal policy $u(t,x)$ and suppose that the nominal trajectory of the system (without noise) starting at some $x_0$ is given by $\{\bar{x}_t\}$ under the nominal control $\{\bar{u}_t\}$. As was done previously, let us now expand the solution of the equation above in terms of the perturbations from this nominal trajectory, $\delta x_t = x_t - \bar{x}_t$. Then, given the optimal nominal control $\bar{u}_t$, we can solve the minimization of the Hamiltonian as:
\begin{align}
    \min_{u_t}&\ H(x_t,u_t) = \min_{\delta u_t} H(\bar{x}_t + \delta x_t, \bar{u}_t + \delta u_t),\\
    & = \min_{\delta u_t} \Big[l(\bar{x}_t + \delta x_t) + \frac{r}{2} \bar{u}_t^2 + (f(\bar{x}_t + \delta x_t) + g\bar{u}_t) \frac{\partial J}{\partial x} \nonumber \\
    & + r \bar{u}_t \delta u_t + \frac{r}{2} \delta u_t^2 + g\delta u_t \frac{\partial J}{\partial x} \Big], \nonumber
\end{align}
which leads to the necessary condition for a minimum:
\begin{align} \label{eq.delt_u_necessary_cond}
(g\frac{\partial J}{\partial x} + r\bar{u}_t ) + r \delta u_t = 0,
\end{align}
which is also sufficient for a minimum since $r>0$ leading to $H$ being strictly quadratic in the variable $\delta u_t$. From Eq.~\eqref{eq.delt_u_necessary_cond}, the optimizing perturbation control is given by $\delta u_t = -\bar{u}_t - \frac{g}{r} \frac{\partial J}{\partial x}$.

Now, let us expand the dynamics and the optimal cost function in the HJB in terms of their perturbations from the nominal trajectory:
$f(x_t) = f(\bar{x}_t) + F_t^1 \delta x_t + \frac{1}{2} F_t^2 \delta x_t^2 + \cdots$, $J(t,x_t) =  \bar{J}_t(\bar{x}_t) + K_t^1 \delta x_t + \frac{1}{2} K_t^2 \delta x_t^2+ \cdots$, where the $F_t^i, K_t^i$ represent the Taylor coefficients of the series expansion of these functions. Therefore, $\frac{\partial J}{\partial x} = K_t^1 + K_t^2 \delta x_t + \frac{K_t^3}{2} \delta x_t^2 + \cdots$, $\frac{\partial^2 J}{\partial x^2} = K_t^2 + K_t^3 \delta x_t + \frac{1}{2} K_t^4 \delta x_t^2+ \cdots$. Noting that the variable $x_t = \bar{x}_t + \delta x_t$, i.e., the space variable has an explicit time dependence via the nominal trajectory, it follows that: 
\begin{comment}
\begin{align}
\frac{\partial J}{\partial t} = \frac{\partial J}{\partial t} + \frac{\partial J}{\partial x}|_{\bar{x}_t} \dot{\bar{x}} \nonumber\\
= [\dot{J}_t(\bar{x}_t) + \dot{K}_t^1 \delta x_t + \frac{1}{2} \dot{K}_t^2 \delta x_t^2+ \cdots] \nonumber\\
+ \dot{\bar{x}}[K_t^1 + K_t^2 \delta x_t + \frac{K_t^3}{2} \delta x_t^2 + \cdots],
\end{align}
\end{comment}
\begin{align} \label{eq.hjb_lhs}
\frac{\partial J(t,x_t)}{\partial t} 
= [\dot{\bar{J}}_t(\bar{x}_t) + \dot{K}_t^1 \delta x_t + \frac{1}{2} \dot{K}_t^2 \delta x_t^2+ \cdots] \nonumber\\
- \dot{\bar{x}}_t[K_t^1 + K_t^2 \delta x_t + \frac{K_t^3}{2} \delta x_t^2 + \cdots],
\end{align}
where, $\dot{\bar{J}}_t(\bar{x}_t), \dot{K}_t^1, \cdots$, are total derivatives with respect to $t$, since they only depend on the time.
\begin{comment}
where the first line above follows from the fact that $\frac{\partial \Psi(x_1,x_2)} {\partial t} = \frac{\partial \Psi}{\partial x_1} \frac{\partial{x}_1}{\partial t} + \frac{\partial \Psi}{\partial x_2} \frac{\partial{x}_2}{\partial t}$, where note above that $x_1 = t$, and $x_2 = \bar{x}_t + \delta x_t$.Thus, the partial derivatives $\frac{\partial J}{\partial t}$ on the two sides of the equality in the first line of the equation above are different.
\end{comment}

Then, using the above expressions, one can express the minimum value of the Hamiltonian in terms of the state perturbations $\delta x_t$ as:
\begin{comment}
\begin{align*}
    &\min_{u_t} H(x_t,u_t) = (f(\bar{x}_t) + F_t^1 \delta x_t + \frac{F_t^2}{2} \delta x_t^2+ \cdots)\\
    & (K_t^1 + K_t^2 \delta x_t + \frac{K_t^3}{2}\delta x_t^2 + \cdots)\\
    &+ g(\bar{u}_t)(K_t^1 + K_t^2 \delta x_t + \frac{K_t^3}{2}\delta x_t^2 + \cdots)\\
    &- \frac{g}{r} [g(K_t^1 + K_t^2 \delta x_t + \frac{K_t^3}{2}\delta x_t^2 + \cdots)\\ 
    &+ r\bar{u}_t](K_t^1 + K_t^2 \delta x_t + \frac{K_t^3}{2}\delta x_t^2 + \cdots) \\
    &+ (l(\bar{x}_t) + L_t^1 \delta x_t + \frac{L_t^2}{2} \delta x_t^2+ \cdots) + \frac{1}{2} r \bar{u}_t^2 \\
    &+ r \bar{u}_t \frac{-1}{r} [g(K_t^1 + K_t^2 \delta x_t + \frac{K_t^3}{2}\delta x_t^2 + \cdots) + r\bar{u}_t]\\ 
    &+ \frac{1}{2r} [g(K_t^1 + K_t^2 \delta x_t + \frac{K_t^3}{2}\delta x_t^2 + \cdots) + r\bar{u}_t]^2.
\end{align*}
\end{comment}
\begin{align} \label{eq.hjb_rhs}
    \min_{u_t} H(x_t,u_t) &= [l(\bar{x}_t) + L_t^1 \delta x_t + \frac{L_t^2}{2} \delta x_t^2+ \cdots] \nonumber\\& 
    - \frac{g^2}{2r} [K_t^1 + K_t^2 \delta x_t + \frac{K_t^3}{2}\delta x_t^2 + \cdots]^2 \nonumber\\
    &+ (f(\bar{x}_t) + F_t^1 \delta x_t + \frac{F_t^2}{2} \delta x_t^2+ \cdots)(K_t^1 \nonumber\\
    & + K_t^2 \delta x_t + \frac{K_t^3}{2}\delta x_t^2 + \cdots).
\end{align}

Next, noting that $\dot{\bar{x}} = f(\bar{x}_t) + g \bar{u}_t$, we obtain the following equations for the evolution of the Taylor co-efficient of the optimal cost function by equating the different powers of $\delta x_t$ on both sides of the stochastic HJB (Eq.~\eqref{eq.stochastic_hjb}) given in Eq.~\eqref{eq.hjb_lhs} and~\eqref{eq.hjb_rhs}.
\begin{align}
    -\dot{\bar{J}}_t &= \bar{l}_t - \frac{g}{r} (\frac{g K_t^1}{2}+r\bar{u}_t)K_t^1 + \epsilon^2 K_t^2, \\
    - \dot{K}_t^1 &= L_t^1 + F_t^1K_t^1 
    - \frac{g}{r}(gK_t^1 + r\bar{u}_t)K_t^2 
    + \epsilon^2 K_t^3, \\
    -\dot{K}_t^2 &= L_t^2 + 2F_t^1 K_t^2 + F_t^2K_t^1 -\frac{g^2}{r}(K_t^2)^2 \nonumber\\
    &-\frac{g}{r}(gK_t^1 + r\bar{u}_t)K_t^3 
    + \epsilon^2 K_t^4,
\end{align}
where we have expanded the first three terms of the expansion in the equations above, and similar expansions may be done for the higher order terms as well.
At this point, we make the following remarks regarding the perturbation expansion above.\\

\begin{remark}
\textbf{Computational Intractability of the Stochastic Problem.} The equations above show that the lower order terms in the stochastic problem are affected by the higher order terms unlike in the deterministic case. Thus, in order to compute the stochastic law, we have to approximate to a high enough order to ensure accuracy in the solution, which in turn implies that the solution of the stochastic problem is very prone to errors. To see this, note that if we were to expand the solution to the $n^{th}$ order, the Taylor co-efficient $K_t^n$  would be affected by the coefficients $K_t^{n+1}$ and $K_t^{n+2}$, and therefore these higher order coefficients would need to be sufficiently small for the resulting solution to be accurate. However, if one approximates to a very high order $n$, quite apart from the obvious curse of dimensionality issue, the resulting system of equations becomes severely ill-conditioned, and consequently, highly sensitive to small errors in the data. Please see our related paper \cite{RL_conv} for the relevant details on this aspect.
\end{remark}
\begin{remark}
\textbf{The Deterministic Problem.} The expressions above also allow us to find the perturbation expansions for the deterministic problem. It is key to note that if the problem considered is deterministic, then $gK_t^1 + r\bar{u}_t=0$ due to the minimum principle, and since $\epsilon = 0$ in the deterministic problem, we obtain the expressions that we derived via the Method of Characteristics in the previous section. The Method of Characteristics is still necessary since it allows us to establish the uniqueness of the optimal nominal trajectory $(\bar{x}_t, \bar{u}_t)$. Thus, the above development can be thought of as an alternative way to derive the perturbation expansion result. Furthermore, we can see that if we are required to derive the cost-to-go of the deterministic policy when applied to the stochastic system, albeit $gK_t^1 + r\bar{u}_t = 0$ due to optimality, nonetheless, there is coupling from the higher order terms due to stochasticity arising from the $O(\epsilon^2)$ terms above, and thus, even this case is intractable to compute. However, since we are interested only in the deterministic feedback law, such a computation is unnecessary.
\end{remark}

\section{\uppercase{NEar Optimality of Shrinking Horizon Model Predictive Control}}\label{sec:4}
In our developments till this point, we have shown that the deterministic feedback law is near-optimal with respect to the optimal stochastic law and that it has a perturbation structure that is lost in the stochastic problem. However, solving the deterministic DP problem is also subject to the Curse of Dimensionality. Nonetheless, owing to the perturbation structure, one can solve the deterministic problem locally (up to the linear feedback term), and then replan at fixed decision time epochs, assuming that the time between the decision epochs is small enough that the local feedback law remains valid in between the epochs. 
Thus, consider a Model Predictive type approach to solving the stochastic control problem. We outline the algorithmic procedure in Algorithm~\ref{algo.MPC-SH} to highlight that our advocated procedure is slightly different from the traditional MPC approach studied in the literature \cite{Mayne_1,Mayne_2}.
\begin{algorithm}[!htbp]
\caption{Shrinking Horizon MPC (MPC-SH)}
\label{algo.MPC-SH}
\textit{Given:} initial state $x_0$, time horizon $T$, cost $c(x,u) = l(x)+ \frac{1}{2}\tr{u}Ru$, terminal cost $c_T(x)$, and decision epoch time $\Delta$.
\begin{algorithmic}[1]
\State Set $N=\frac{T}{\Delta}$, $x_i = x_0$.
\While{$t < N \Delta$} 
\State Solve the open-loop (noise-free) optimal control problem for initial state $x_i$, along with the associated linear perturbation feedback, for the horizon ($N\Delta - t$). Let the perturbation feedback law be denoted by $u(t,x) = \bar{u}_t + K_t \delta x_t$, where $\delta x_t = x_t - \bar{x}_t$ and $(\bar{x}_t, \bar{u}_t)$ is the optimal nominal trajectory.
\State Apply the perturbation feedback law $u(t,x)$ till time $(t+\Delta)$ and observe the state $x_f = x_{(t+\Delta)}$.
\State Set $t = t + \Delta$, $x_i = x_f$.
\EndWhile
\end{algorithmic}
\end{algorithm}

\begin{remark}
In traditional MPC \cite{Mayne_1,Mayne_2}, the horizon $N$ to solve the open-loop problem is fixed. The setting is deterministic, and the necessity of replanning for the problem stems from the assumption that the actual problem horizon is infinite, and therefore, computationally intractable. In lieu, our problem horizon is finite, the repeated replanning takes place over progressively shorter horizons, and the need for replanning arises from the stochasticity of the problem. In particular, note that if the system were really deterministic, there would be no need for replanning.  %Albeit the shrinking horizon seems like a ``trivial" change, it makes a highly significant difference in practice, for instance, see our empirical results in Section V B and V C.
\end{remark}
\begin{theorem} \textit{Near-Optimality of MPC-SH.} \label{theorem.MPC}
The MPC feedback policy obtained from the application of the \textit{Shrinking Horizon MPC} algorithm is near-optimal to $O(\epsilon^4)$ to the optimal stochastic feedback policy for the stochastic system \eqref{eq:model}.
\end{theorem}
\begin{proof}
%For simplicity, we show a sketch here, the detailed proof is in the Appendix.\\
We know that $J^0 (t, x) = \varphi^0 (t,x)$, and $J^1(t,x) = \varphi^1(t,x)$ from Proposition \ref{prop_e4}, for all $(t,x)$. Owing to the uniqueness and global optimality of the open-loop from Proposition \ref{OL_optimality}, it follows that the nominal control sequence, and the associated linear perturbation feedback law, found by the MPC procedure outlined above coincides locally with the optimal deterministic feedback law given any state $x$ and any time $t$. Therefore, the result follows.
\end{proof}
Note that the proof above also shows that the MPC-SH procedure provides the optimal deterministic feedback law which is stated in the following corollary.
\begin{corollary}
    The MPC-SH algorithm provides the optimal deterministic feedback law given any initial condition. 
\end{corollary}

Further, note that the linear deterministic perturbation feedback is also near-optimal to fourth order owing to the fact that the terms till $O(\epsilon^2)$ are solely due to the linear part of the perturbation feedback (see Remark \ref{remark:linear_fb_nearoptimality}). This leads to the following result.
\begin{corollary}
    The deterministic linear perturbation feedback policy (T-PFC): $\pi_t^{d,l}(x_t) = \bar{u}_t + K_t \delta x_t $ where $K_t$ is defined in Eq. \eqref{T-PFC-K}, is near optimal to fourth order in $\epsilon$ for the stochastic system \eqref{eq:model}.
\end{corollary}

The results above establish that MPC-SH and the deterministic linear perturbation feedback are good approximations to the optimal stochastic policy for low noise levels. Now, we examine two particularly important consequences of the results above for Stochastic MPC and RL.

\textbf{Stochastic MPC}. {%The MPC-SH procedure we propose is different from the traditional MPC (fixed horizon MPC) because of our shrinking horizon whereas typically MPC considers as a fixed horizon $H$ that is small compared to the actual horizon $T$, i.e., $H \ll T$. %Albeit the stability of the traditional MPC approach can be established, the actual domain of attraction of such a policy is limited by the fact that the MPC policy needs to get the system into a control invariant subset containing the origin in at most $H$ steps \cite{Mayne_1}. %In practice, these policies tend to be sluggish and have difficulty in getting to the goal (please see Figures~\ref{fig:varyingHorizonMPC} and~\ref{fig:trajectories}). Thus, in our opinion, at least our preliminary results indicate that we should consider a problem with a finite horizon and shrink the horizon, since it is a better approximation of the real problem and results in much better performance. However, we note that a much more careful investigation needs to be done of this issue and is left for future work.}\\
A major computational bottleneck with stochastic MPC \cite{Mayne_1}, is that the MPC search needs to be over (time-varying) feedback policies rather than control sequences owing to the stochasticity of the problem, which leads to an intractable optimization for nonlinear systems. 
Because of this intractability, most of the work in stochastic MPC deal with linear systems using stochastic tube approach \cite{smpc_mesbah2016, smpc_heirung2018}. % and some more recent work using generalized polynomial chaos (gPC) \cite{kim2013generalised, FISHER2009polychaos}. Nonlinear stochastic MPC using gPC also typically solves over control sequences instead of feedback policies for tractability. 
However, as our results demonstrate, the MPC feedback law we propose (MPC-SH), at least in the absence of constraints, is near-optimal to the fourth order. Further, as we have shown analytically in Section~\ref{sec.3D} and as will be seen from our empirical results in Section~\ref{section:results}, in practice, the solution of the stochastic DP problem, even in the absence of constraints, is highly sensitive to noise, and there is no notion of a local (tube) solution in the nonlinear case, quite apart from the usual issue of dimensionality, and MPC-SH gives much better performance than the solution of the stochastic DP problem.
A further important practical consequence of Theorem 1 is that we can get performance comparable to MPC (see Section 5.2), using T-PFC, and replanning the nominal sequence only when needed, similar to the event driven MPC philosophy \cite{ETMPC1,ETMPC2}.\\
We note here that there is very recent work from the MPC literature that utilizes the same shrinking horizon approach \cite{Diehl-e4} and proves a similar fourth order near optimality result. %However, the results regarding the global optimality of the deterministic open loop problem, the perturbation structure of the deterministic and stochastic problems, the fourth order near optimality of the perturbation feedback control, and their computational implications, are novel to this paper.\\
%This event driven replanning approach, and its computational benefits, are also demonstrated in the next section.\\

\textbf{Reinforcement Learning.} The problems considered in reinforcement learning can be construed as one of finding the optimal feedback policy for a stochastic nonlinear dynamical system \cite{bertsekas1}. Typically, this is done via simulations or rollouts of the dynamical system of interest, which allied with a suitable function approximator such as a (deep) neural net, yields a nonlinear feedback policy. However, these methods tend to be highly data intensive, slow to converge, and suffer from extremely high variance in the solution since they try to solve the DP equation \cite{RL_conv}. This is a manifestation of the inherent curse of dimensionality in trying to solve the stochastic DP problem. In fact, it is much easier to repeatedly solve the open-loop problem as prescribed by MPC. Of course, there remains the problem of whether we can solve the open-loop problem online. In our opinion, this is feasible today, when allied with efficient computational algorithms like ILQR \cite{ILQG_tassa2012synthesis} with suitable modifications. %that exploit the causal structure of optimal control problems, suitable high performance computing (HPC) modifications, and suitable randomization of the computations via rollouts that can help us very efficiently estimate the system parameters involved. 
In fact, this is the subject of the second part of this paper on data-based control \cite{wang2022search}.
\section{Empirical Results}\label{section:results}
This section will show evidence for theoretical results derived previously through computational experiments. In subsection~\ref{sec.5A}, the inaccuracy of the stochastic solution, as discussed in Remark 6, will be shown for a simple 1-D problem in comparison with the deterministic solution. The near-optimality of MPC-SH, which was theoretically shown to be the optimal deterministic solution in Theorem~\ref{theorem.MPC}, will also be compared with the stochastic solution in a nonlinear problem. Further, we will show why it is intractable to solve the stochastic HJB accurately. In Subsection~\ref{sec.5B}, the performance of using the optimal linear perturbation feedback derived in Section~\ref{sec:MOC} will be compared with MPC-SH on nonlinear robotic problems. The experiments shown in this section are carried out over 500 Monte Carlo simulations, and the performance statistics are computed from these simulations.  

\subsection{Deterministic vs. Stochastic policy} \label{sec.5A}
In this section, we aim to show, through simulations, that computing the optimal stochastic feedback law is subject to errors, as explained by the theory discussed previously in Sec.~\ref{sec.3D}. We show this by comparing the performance of the deterministic solution applied to the stochastic problem and the stochastic solution in a linear and a nonlinear problem. We consider the following problem:
\begin{subequations}
\begin{align}
    &J(0,x_0) = \nonumber\\ 
    & \min_{\{u_t\}} \Exp{}{ \frac{1}{2}  \left( \int_{0}^{T} (qx_t^2 + ru_t^2) dt 
                 + q_T x_T^2 \right)}  \\
    &\text{s.t.}~ dx = (f(x) + g(x)u)dt + \epsilon dw, ~\text{given}~ x_0. 
\end{align}
\end{subequations}
The solution to the above problem is calculated by solving the HJB equation (written for the scalar case):
\begin{equation}
    -\frac{\partial J}{\partial t} = \frac{1}{2}  \Big(qx^2 - \frac{g(x)^2}{r}  \Big(\frac{\partial J}{\partial x}\Big)^2\Big) + f(x) \frac{\partial J}{\partial x} + \frac{\epsilon^2 Q}{2} \frac{\partial^2 J}{\partial x^2}, ~\label{eq.hjb_stochastic}
\end{equation}
where, $J = J(t,x)$ is the expected cost-to-go from state $x$ at time $t$, with terminal condition $J(T,x) = \frac{1}{2} q_T x^2$. The minimizing optimal control $u = -\frac{1}{r} g(x) \frac{\partial J}{\partial x}$ and we take $q = 100$, $q_T = 500$  and $r =1$. The noise $w$ added to the system in stochastic cases is zero mean Gaussian white noise, with standard deviation being the maximum value of the control input obtained from the nominal trajectory by solving the deterministic problem - $(\sqrt{Q}= \bar{u}_{max})$. The HJB equation in \eqref{eq.hjb_stochastic} is solved by the finite difference (FD) method in a fixed domain since it is the method of choice for solving advection-diffusion PDEs, which \eqref{eq.hjb_stochastic} is, in the computation fluid dynamics community \cite{FD_CFD}. The parameters used in FD are shown in Table~\ref{table:fd_parameters}. The time and space discretization was chosen to satisfy the Courant–Friedrichs–Lewy (CFL) conditions \cite{CFL}. We consider only a 1-D problem for the sake of easy illustration since \eqref{eq.hjb_stochastic} becomes computationally intractable to solve for high-dimensional problems; nevertheless, these simple low dimensional problems clearly illustrate the issues with solving the stochastic HJB equation.

\begin{table}[!htbp] 
\centering
\begin{tabular}{|c|c|c|c|}
    \hline 
    Domain & $\Delta x$ & $\Delta t$  & $T$\\
    \hline
    $[-2, 2]$ & $0.02 $& $3.33\times 10^{-6}$  & 1\\
     \hline
\end{tabular}
\caption{Parameters used in finite difference solution to HJB PDE.}
\label{table:fd_parameters}
\end{table}

%\subsection{\textbf{System Description.}}
\textbf{Linear Quadratic case:}  
The linear quadratic case is specifically chosen since we know the optimal solution for the stochastic optimal control problem. We consider the following linear system 
\begin{align}
    %\min_{\{u_t\}} & \Exp{}{\frac{1}{2} ( \int_{0}^{1} (qx(t)^2 + ru(t)^2) dt + q_f x(1)^2 )}\\
    dx = (x + u)dt + \epsilon dw, ~\text{given}~ x_0 = 1. \label{eq.1d_dynamics}
\end{align}
In the linear quadratic Gaussian case, we know that the optimal solution for the stochastic problem and the deterministic problem ($\epsilon = 0$) are the same and are given by solving the Riccati equation governing the Hessian matrix of the cost-to-go function. The resulting control policy is the well-known linear quadratic regulator (LQR). In order to validate the accuracy of the HJB finite difference (HJB-FD) solution, we use LQR as a comparison since it is the optimal solution for this case and can be computed accurately. Note that the stochastic HJB is still a second-order nonlinear PDE, even in the linear-quadratic case. 

% The expected cost-to-go of the LQR policy in the stochastic case for our case is given by these standard equations
% \begin{align}
%     \Exp{}{J_{lqr}(x,t)} = \frac{1}{2} p_t x^2 + v_t, \\
%     -\dot{p}_t = 2p_t - \frac{p_t^2}{r} + q, ~p_T = q_T, \label{eq.lqr_p} \\
%     -\dot{v}_t = p_t\sigma^2, ~q_T = 0. 
% \end{align}

In our first experiment, we solve the HJB equation in~\eqref{eq.hjb_stochastic} for a particular value of $\epsilon$ in the domain $[-2,2]$. We use the obtained feedback policy $u = -\frac{1}{r} \frac{\partial J}{\partial x}$, and apply it to the linear system given in \eqref{eq.1d_dynamics} using the same $\epsilon$ value the HJB was solved for, to regulate the noise acting on the system. The linear system is simulated under this feedback policy for the initial condition $x_0 = 1$, for a time interval of $[0,1]$. We do Monte Carlo simulations of the system for different noise samples of $w$ and obtain the mean and standard deviation of the cost incurred by the system over these experiments. We do the same experiment for LQR with the LQR's feedback policy $u_t = - K^{lqr}_t x_t $ and find the mean and standard deviation of the cost incurred by the system from the given initial state. We repeat the experiment for different values of $\epsilon$ and present the results in Fig.~\ref{fig.1d_linear}.  From Fig.~\ref{fig.1d_linear}, it is clear that both the solutions only match for low noise cases, and in the high noise case, the stochastic HJB-FD solution is highly inaccurate. 

\begin{figure}[!htbp]
    
    \subfloat[Full noise spectrum.]{\includegraphics[width=0.5\columnwidth]{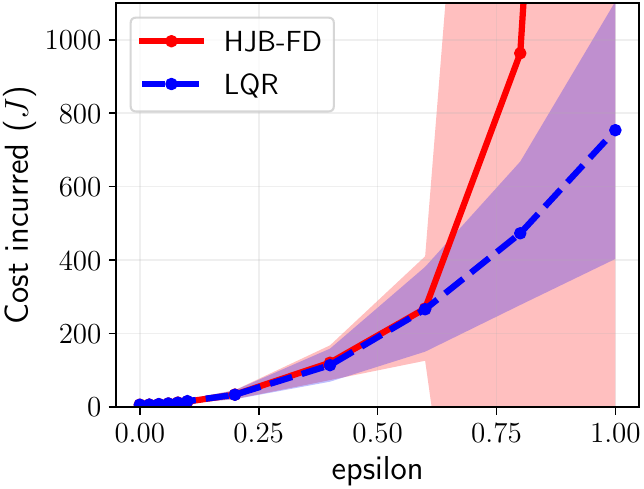}}
    \hfill
    \subfloat[Low noise region enhanced.]{\includegraphics[width=0.48\columnwidth]{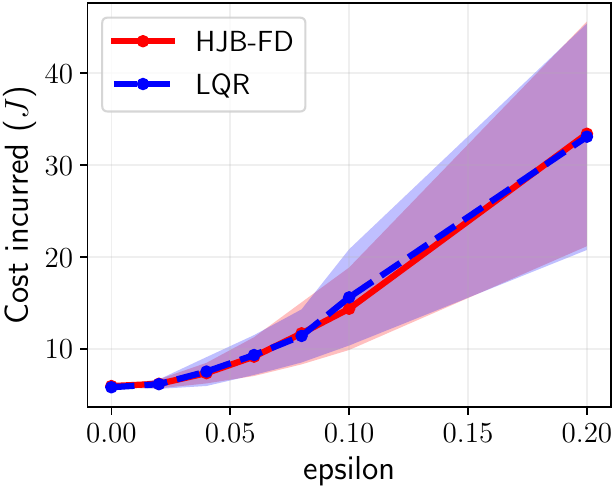}}
    \caption{Performance comparison of the solution to HJB (Eq.~\eqref{eq.hjb_stochastic}) obtained from finite difference scheme and LQR on the linear system. The mean and standard deviation of the cost incurred by the system are calculated from experiments for different cases of $\epsilon$. The data shown is from simulating the linear system from initial condition of $x_0 = 1$.}
    \label{fig.1d_linear}
\end{figure}

In the following, we show the reason behind the inaccuracy of the HJB-FD solution. We first compare the expected cost-to-go values between the two methods for different cases of $\epsilon$. The expected cost-to-go for HJB is directly available from the solution to \eqref{eq.hjb_stochastic} over the entire domain. For LQR, we can calculate the expected cost-to-go for the stochastic case by solving the following standard equations (written specifically for our linear system in \eqref{eq.1d_dynamics}):
\begin{align}
    J_{lqr}(t,x) = \frac{1}{2} p_t x^2 + v_t, \\
    -\dot{p}_t = 2p_t - \frac{p_t^2}{r} + q, ~p_T = q_T, \label{eq.lqr_p} \\
    -\dot{v}_t = p_t (\epsilon \sigma)^2, ~v_T = 0. 
\end{align}
We show, in the left column of Fig.~\ref{fig.CTG_slices}, the cost-to-go function of the HJB solution and the LQR solution in our domain at a specific time instant. We observe that the HJB-FD cost-to-go exactly matches the LQR for $\epsilon = 0$ and $0.2 $. For $\epsilon = 1$, it can be seen that the cost-to-go values don't match and we can infer that the HJB-FD solution is inaccurate. Since the feedback policy depends on the cost-to-go function, it will also be inaccurate in such cases. 

The reason for the inaccuracy of the stochastic HJB-FD solution is illustrated in the right column of Fig.~\ref{fig.CTG_slices}. The plots in the right column show the trajectories taken by the system under the HJB-FD feedback policy for different values of $\epsilon$. When the $\epsilon$ value parametrizing the strength of the noise becomes large, it can be seen that the trajectories leave the domain on which the solution is obtained, due to the noise acting on the system. Since the cost-to-go solution is unavailable outside the domain, one has to approximate the cost of these trajectories with the cost at the boundary. To get an accurate solution, the domain one has to solve needs to expand with time. Since most computational methods do a fixed domain approximation, the stochastic solution obtained will inherently be inaccurate because the states inside the boundary need the cost-to-go values of states outside the domain as the noise intensity increases. In the deterministic case, owing to the absence of noise, the control takes the system towards the origin and not outside the domain. So, the deterministic cost-to-go and feedback policy is always accurate. Furthermore, note that when the stochastic HJB solution is accurate, the system does not leave the domain owing to the control dominating the effect of the noise.

\begin{figure}[!htbp]
    \centering
   \subfloat[ $\epsilon = 0$]{\includegraphics[width=.49\linewidth]{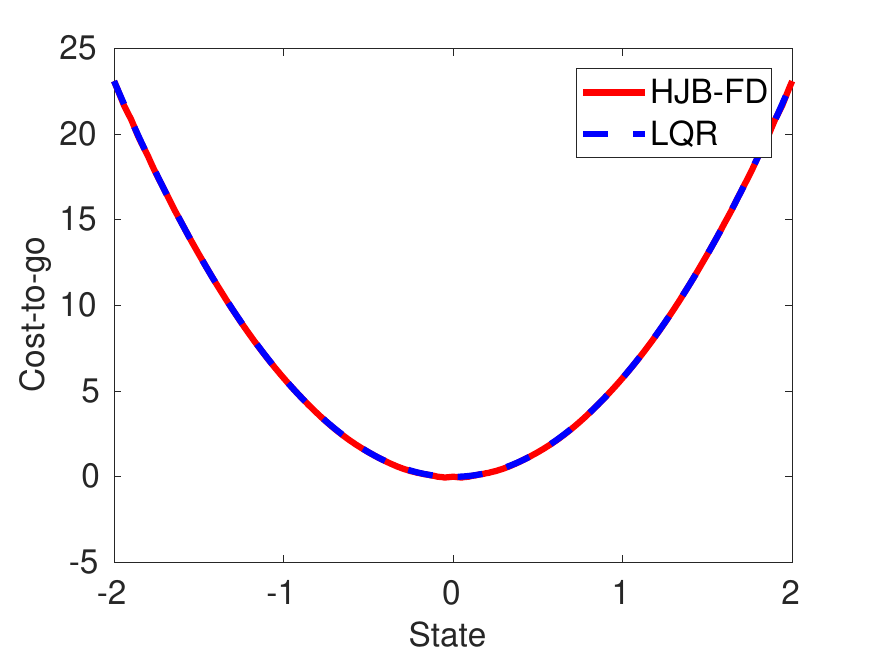}}
   \subfloat[$\epsilon = 0$]{\includegraphics[width=.49\linewidth]{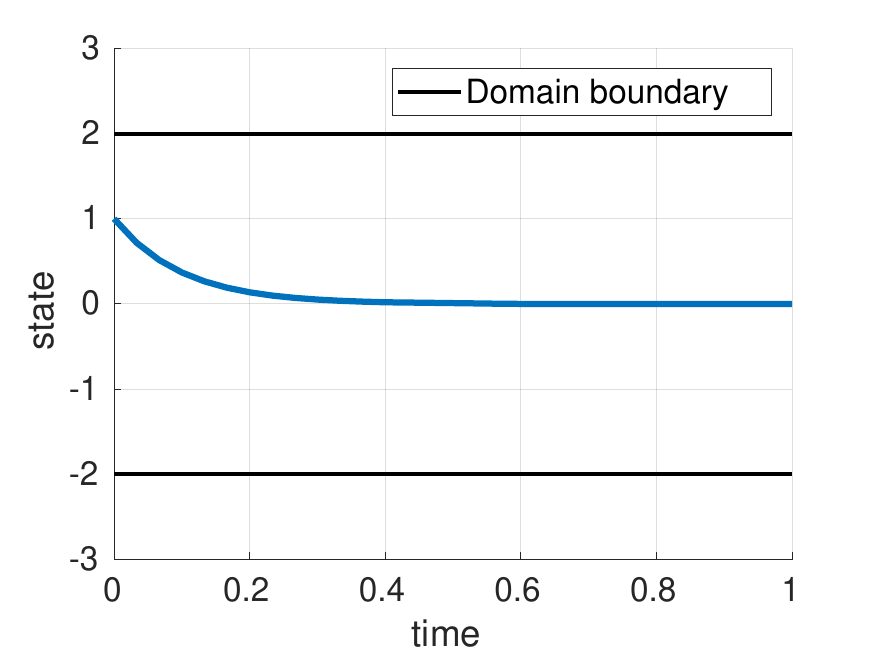}}
   \newline
    \subfloat[$\epsilon = 0.2$] {\includegraphics[width=0.49\linewidth]{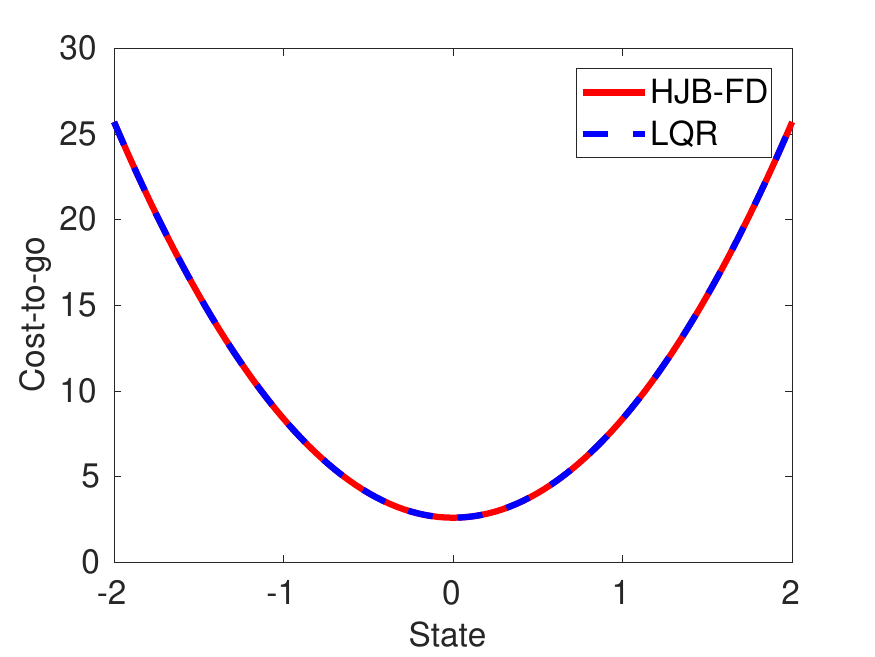}}%
    \subfloat[$\epsilon = 0.2$]{\includegraphics[width=.49\linewidth]{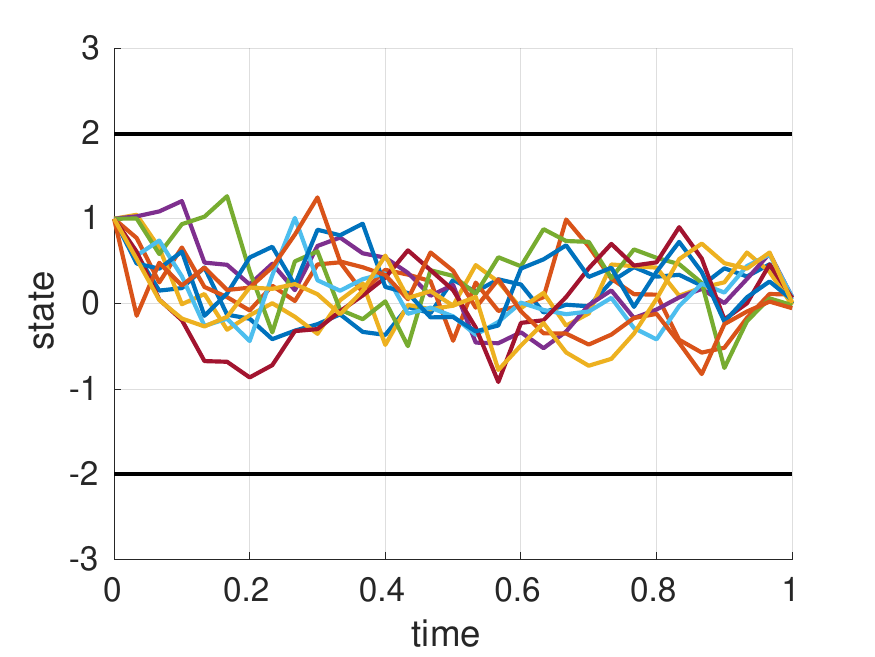}}
    \newline
    \subfloat[$\epsilon = 1$] {\includegraphics[width=0.49\linewidth]{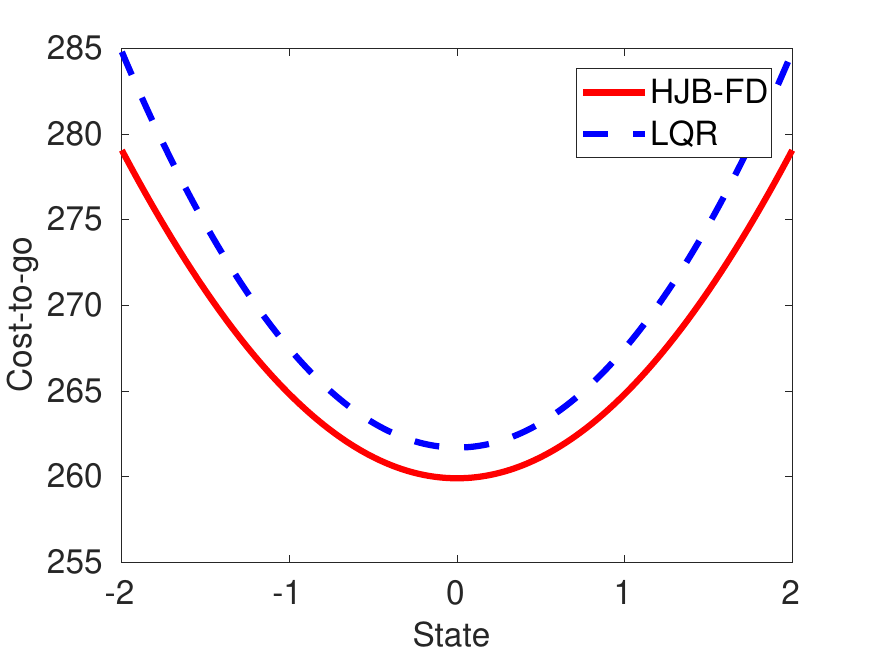}}%
    \subfloat[$\epsilon = 1$]{\includegraphics[width=.49\linewidth]{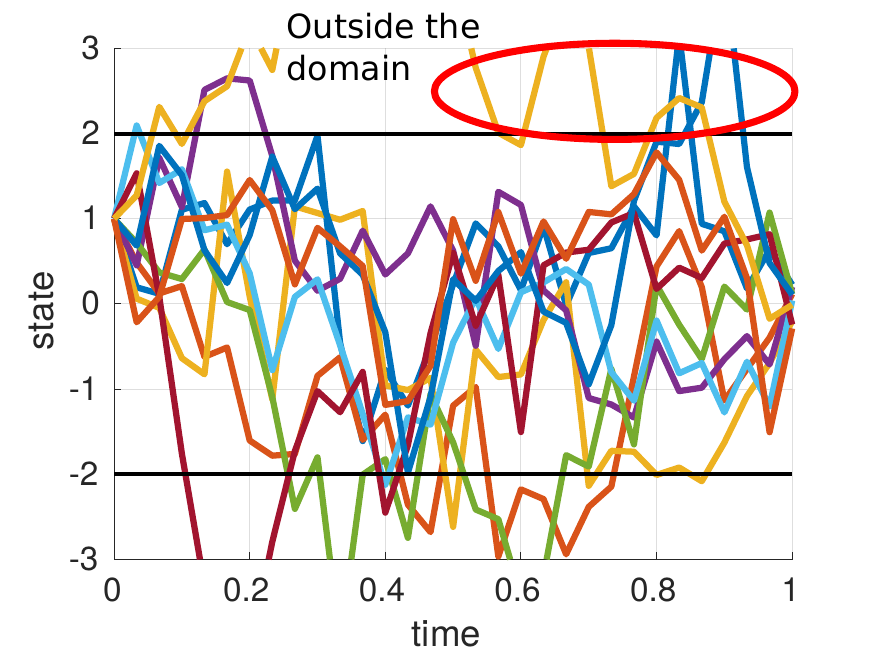}}
    \caption{Left column: Comparison of expected cost-to-go value from the HJB-FD solution and the LQR Riccati solution at $t = 0.81$. The plot shows that the HJB-FD cost-to-go doesn't match the LQR cost-to-go (which is the optimal) at high noise levels. \\
    Right column:
    Sample trajectories of the linear system at different noise levels under the policy computed by HJB-FD. Since the trajectory could leave the domain in high noise cases, the expected cost-to-go calculated in seeking the stochastic feedback policy will be inaccurate. (The trajectories were generated with the original sampling time of the FD solver, but the data is plotted at a larger sampling interval for the sake of clarity.) }
    \label{fig.CTG_slices}
\end{figure}

\textbf{Nonlinear Case.}
We consider the nonlinear system $dx  = (-cos(x) + u)dt + \epsilon dw $ with initial condition $x_0 = 1$.
As discussed in Sec.~\ref{sec:4}, MPC-SH feedback law is the optimal feedback law for the deterministic problem and the cost is $O(\epsilon^4)$ near-optimal to the stochastic cost. The algorithm for MPC-SH is given in Algorithm 1. To solve the open-loop optimization problem in MPC-SH, the iterative linear quadratic regulator (ILQR) algorithm is used \cite{ILQG_tassa2012synthesis}. ILQR is used specifically since the converged optimal solution satisfies the necessary conditions of the minimum principle given in Eqs.~\eqref{L-C-1},~\eqref{L-C-2}. As discussed in Proposition~\ref{OL_optimality}, the deterministic open-loop problem has a unique minimum for our case, and ILQR will guarantee convergence to it \cite{wang2022search}.  

The HJB equation is solved using FD, as discussed in the linear-quadratic case. The open-loop optimization in MPC-SH is solved using ILQR as discussed above for the specific initial condition and tested on the stochastic nonlinear system for a value of $\epsilon$. The experiment is repeated for different noise levels by varying $\epsilon$. The decision epoch time chosen for MPC-SH was $\Delta = 0.005$, approximately $1000 \times$ the $\Delta t$ used in FD. The mean and standard deviation of the cost incurred in these experiments are tabulated in Fig.~\ref{fig.1d_mpc}. 
Fig.~\ref{fig.1d_mpc} shows that the MPC-SH feedback law has comparable performance with the stochastic HJB-FD solution. MPC-SH is also computationally more efficient to solve, as HJB-FD requires very fine time discretization to solve without numerical issues even for the 1-D case owing to the CFL conditions (see table~\ref{table:fd_parameters}). \nxx{Also, MPC-SH finds an optimal trajectory for a single initial condition as opposed to HJB-FD which finds all solutions over the entire domain, which is computationally expensive.}
\begin{figure}[!htbp]
    \centering
   \subfloat[Full noise spectrum.]{\includegraphics[width=.5\linewidth]{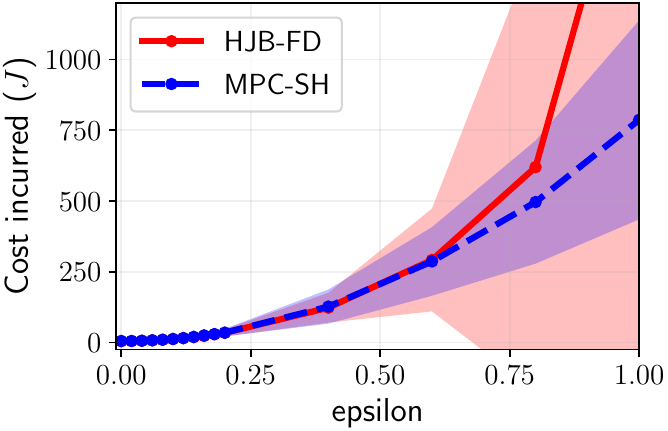}}
   \hfill
    \subfloat[Low noise region enhanced.] {\includegraphics[width=0.48\linewidth]{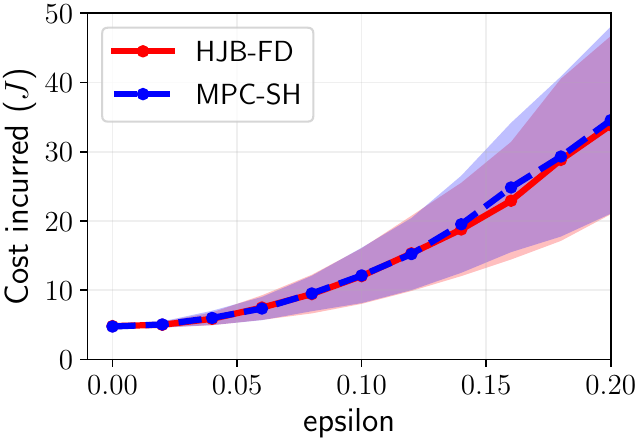}}%
    \caption{Performance comparison of HJB-FD and MPC-SH on the 1-D nonlinear system for different noise levels.}
    \label{fig.1d_mpc}
\end{figure}
\begin{figure}[!htbp]
    \centering
   \subfloat[HJB-FD cost-to-go.]{\includegraphics[width=.5\linewidth]{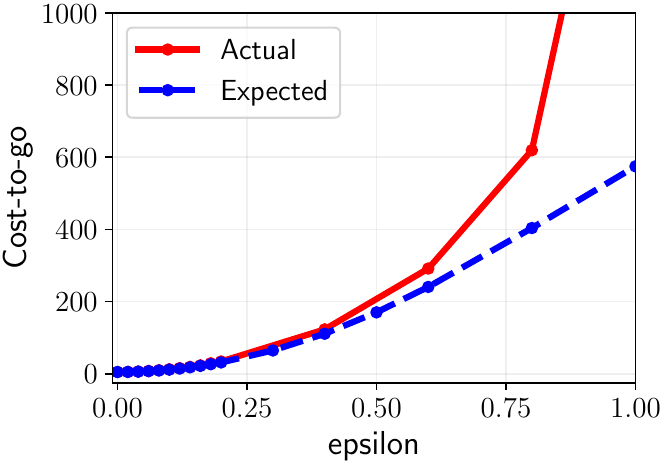}\label{fig.HJB_FD_expected_subplot}}
   \subfloat[MPC-SH $\epsilon = 0.8 $]{\includegraphics[width=.5\columnwidth]{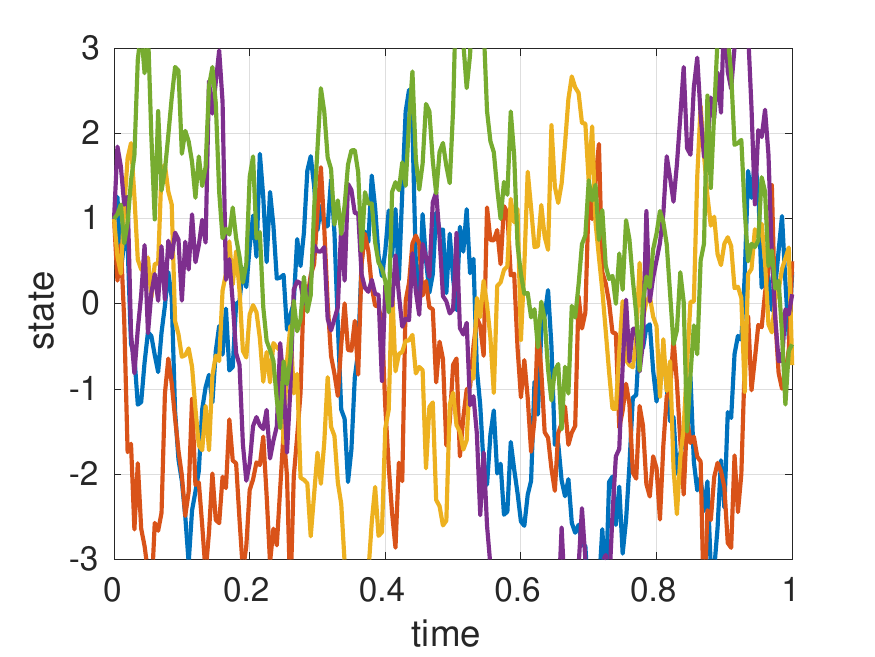} \label{fig.mpc_1dcos_traj8}}
   \newline
   \subfloat[HJB-FD $\epsilon = 0.2 $]{\includegraphics[width=.5\columnwidth]{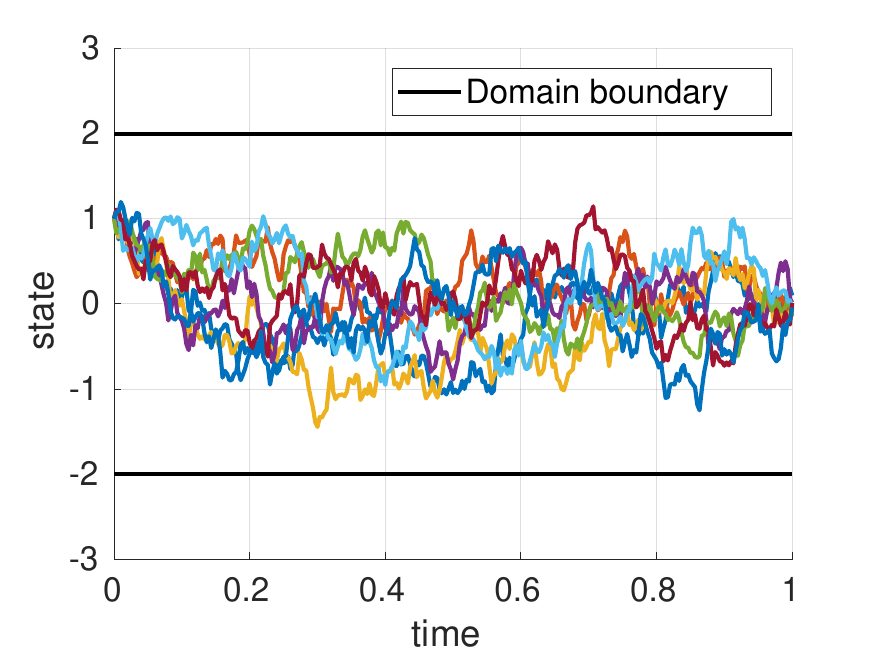} \label{fig.1dcos_traj2}}
   \subfloat[HJB-FD $\epsilon = 0.8 $]{\includegraphics[width=.5\columnwidth]{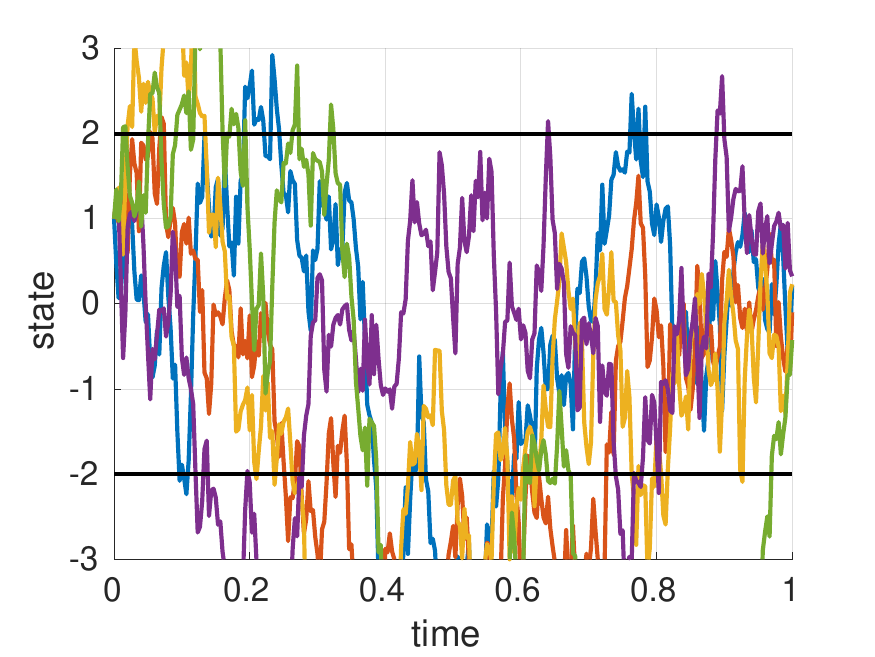} \label{fig.1dcos_traj8}}
    \caption{(a) Comparison of the expected cost-to-go obtained from the HJB-FD solution and the actual cost incurred by applying the HJB-FD feedback policy on the nonlinear system. The cost-to-go is obtained for the initial condition $x_0 = 1$ and the actual cost is the average cost of 500 simulations. Trajectory samples of the nonlinear system under the MPC-SH policy are shown in (b), and under HJB-FD policy for two different cases of $\epsilon$ are shown in (c) and (d).}
    \label{fig.HJB_FD_expected}
\end{figure}

Even when the deterministic solution, which MPC-SH is, is applied to the stochastic case, the performance is almost equivalent, due to the $O(\epsilon^4)$ near-optimality of the deterministic solution to the stochastic. Moreover, the stochastic policy has higher variance than the deterministic MPC-SH policy at $\epsilon = 0.8$, and fails after that - another case that shows that the calculated stochastic policy is inaccurate. To illustrate the inaccuracy in the HJB-FD solution, we compare the expected cost-to-go value calculated by solving the HJB with the true cost of operation in Fig.~\ref{fig.HJB_FD_expected_subplot}. It can be seen that the cost-to-go becomes inaccurate after $\epsilon = 0.6$. As discussed earlier, one has to expand the domain to calculate the cost-to-go of the stochastic problem accurately. Expanding the domain makes the problem more computationally expensive, and trajectories will still leave the domain in high noise cases. In contrast, MPC-SH does not face the issue of computational inaccuracy when a trajectory exits the boundary since it can compute a new trajectory from any given state without worrying about the boundary and the boundary conditions as required by HJB-FD. In particular, this may be construed as the primary computational benefit of using the MPC-SH approach.
\begin{comment}
\begin{figure}
    \centering
    \subfloat[$\epsilon = 0$]{\includegraphics[width=.49\linewidth]{Figures_revision/traj_epsi0.pdf}}
    \subfloat[$\epsilon = 0.2$]{\includegraphics[width=.49\linewidth]{Figures_revision/traj_epsi2.pdf}}
    \newline
    \centering
    \subfloat[$\epsilon = 0.5$]{\includegraphics[width=.49\linewidth]{Figures_revision/traj_epsi5.pdf}}
    \caption{Trajectory samples of the system at different noise levels under the policy computed by HJB-FD. The HJB-FD is solved in the domain $[-2,2]$. Since the trajectory could leave the domain in high noise cases, the expected cost-to-go calculated in seeking the stochastic feedback policy will be inaccurate. \textcolor{red}{What is the sampling time that you use when you plot the trajectories above? Is it the same as the sampling time you use to solve the HJB?}}
    \label{fig:boundary_effect}
\end{figure}
\end{comment}
%
\subsection{Comparison between MPC-SH and T-PFC} \label{sec.5B}
In this section, we will show the comparison in performance of two different deterministic feedback laws: the optimal linear feedback and the MPC-SH feedback law. In Remark~\ref{remark:linear_fb_nearoptimality}, it was shown that the optimal linear feedback controller given by Eqs.~\eqref{T-PFC-G}-\eqref{T-PFC-K}, designed around the optimal open-loop nominal trajectory is also near-optimal to the order of $O(\epsilon^4)$ to the stochastic system. This design is referred to as the trajectory-optimized perturbation feedback controller (T-PFC) \cite{parunandi2019TPFC}. The difference between T-PFC and MPC-SH is that, T-PFC plans the nominal trajectory only once, from the initial state, and uses the linear feedback to correct for errors during its execution. While, MPC-SH replans the nominal trajectory from the current state continuously and uses the linear feedback only for a short interval $\Delta$ between the replans. The advantage of using T-PFC is that the open-loop optimization has to be carried out only once (preferably offline), and the precomputed linear feedback gains can be used to correct for deviations due to uncertainty online. In a stochastic setting, this optimal nominal trajectory generated initially is only optimal if the system stays close to the nominal. If it deviates, the trajectory has to be replanned from the current state as done by MPC-SH to maintain optimal performance. We will examine how the performance of T-PFC compares with MPC-SH in nonlinear robotics problems, namely the car-like robot and cart-pole system, for different noise levels in Fig.~\ref{fig.cost_robotic}. 

In Fig.~\ref{fig.cost_robotic}, we see that T-PFC shows comparable performance to MPC-SH for low values of $\epsilon$. As noise increases, the trajectory deviates from the nominal computed initially, and the feedback policy is no longer optimal, necessitating the need for a replanned nominal trajectory from the current state. Hence, the performance of T-PFC deteriorates for high noise levels. Nevertheless, there is value for T-PFC-like deterministic feedback laws in applications that wish to minimize onboard computing and act in low-noise settings.

\begin{figure}[!htbp]
    \centering
   \subfloat[Car-like robot]{\includegraphics[width=.49\linewidth]{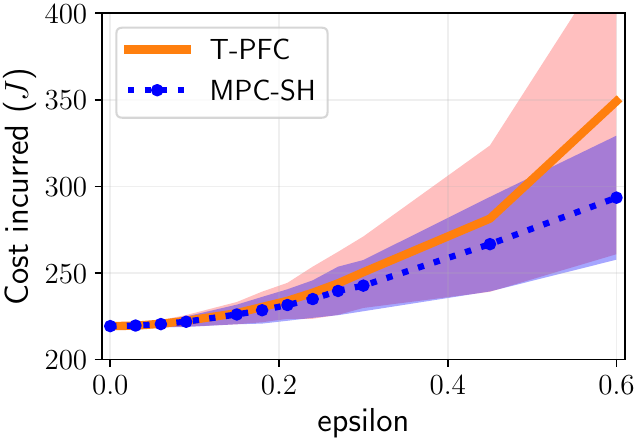} \label{cost_comp}}
    \subfloat[Cartpole] {\includegraphics[width=0.50\linewidth]{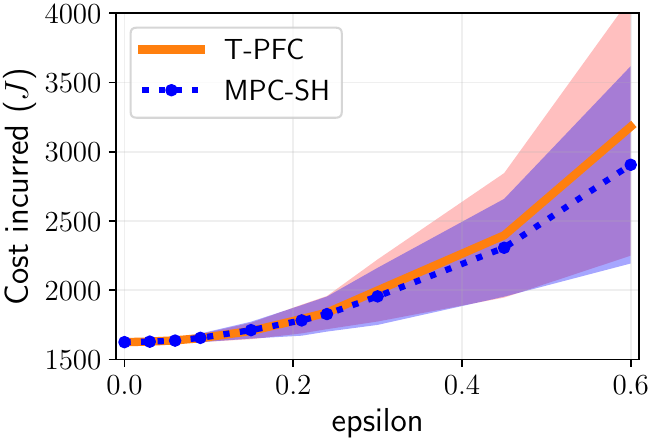}}%
    \caption{Performance comparison of T-PFC with MPC-SH in nonlinear robotics systems. Both policies are computed for a specific initial condition and tested on 500 different samples for each value of $\epsilon$ to find the cost statistics. \nxx{The car-like robot considered is a 4-D system and is governed by the equations $\dot{x} = v cos\theta$, $\dot{y} = v sin\theta$, $\dot{\theta} = \frac{v}{L}tan \phi$, $\dot{\phi}= \omega$, where $v, \omega$ are the control inputs and $L$ is the length of the car. The cart-pole is also a 4-D system and is governed by $(M+m)\ddot{x} - mL \dot{\theta}^2 sin\theta  + mL\ddot{\theta}  cos \theta = F$, $mL^2\ddot{\theta} + mL\ddot{x} cos\theta  +mgL sin\theta = 0$, where $F$ is the control input, and $M, m, L$ are the mass of the cart, mass of the pole and length of the pole. Process noise was added to the above systems after propagating the dynamics at every time step. The standard deviation of the noise added was the maximum value of the states in the optimal nominal trajectory.}}
    \label{fig.cost_robotic}
\end{figure}

\subsection{Discussion} \label{sec.discussion}
The primary takeaway from Section~\ref{sec.5A} and \ref{sec.5B} is that deterministic policies are not only near-optimal but also accurate, scalable, and repeatable. It is not possible to compute the stochastic policy accurately, as shown in Sec.~\ref{sec.5A}. Note that the inaccuracy is not a limitation of the finite difference method used. Other PDE solution techniques like the Finite Element method and pseudo-spectral methods are also solved on a bounded domain, and consequently, not immune to the errors observed in FD. As discussed in Sec.~\ref{sec:4}, random sampling-based methods like approximate dynamic programming, and reinforcement learning are dependent on their samples to explore the domain and inherently have the same issue in the stochastic case. In high dimensions problems, one needs a prohibitively large number of samples to explore the domain. An inefficient sampling of the domain will lead to inaccurate policies as the cost-to-go is not accurately captured by the samples. Due to this issue, there is an inherent variance in the solution obtained by such methods \cite{RL_conv}.
We have done an exhaustive investigation comparing the deterministic feedback approach with other RL methods in the companion paper \cite{wang2022search}, where we report the accuracy, scalability, efficiency, and repeatability of the deterministic policy that the stochastic RL methods lack. 
To summarize, as shown in Fig.~\ref{fig.HJB_FD_expected}, the regime where the stochastic solutions can be computed accurately is the one of low noise where the deterministic solution gives near-identical performance, and consequently, in practice, the deterministic feedback law implemented via MPC-SH is sufficient. %In fact, this result was anticipated by Wonham in an early paper on optimal stochastic control \cite{wonham}: ``\textit{Since the mathematical model is usually greatly complicated by explicitly including
\section{{Conclusion}}
\label{section:conclusion}

In this paper, we have considered the problem of stochastic nonlinear control. We have shown that recursively solving the deterministic optimal control problem from the current state, \`a la MPC, results in a near-optimum policy to fourth order in a small noise parameter, and in practice, empirical evidence shows that the MPC law performs better than the law obtained by computationally solving the stochastic DP problem owing to the perturbation structure of the deterministic optimal control problem. An important limitation of the current work is the smoothness of the HJB solution such that suitable Taylor expansions are possible which may break down in the presence of constraints and one then needs the generalized concept of a viscosity solution \cite{HJB_viscosity}. It remains to be seen if, and how, one may extend the results presented here to such applications.  Also, a careful investigation into the relative merits and demerits of the shrinking horizon approach to MPC when compared to the traditional fixed horizon approach is required for nonlinear problems, as is the generalization to the practical and important partially observed problem.

% \documentclass[twoside]{article}

%\onecolumn

\begin{appendix}
\setcounter{section}{0}
\section{DETAILED PROOFS OF RESULTS}

% The supplementary materials contain detailed proofs of the results that are missing in the main paper.

\textbf{Proof of Lemma 1.}\\
\begin{proof} For ease of notation, we prove the scalar case, the vector case is a straightforward generalization. We proceed by induction. The first general instance of the recursion occurs at $k=3$.
It can be shown that: 
 \begin{align*}
 \delta x_3 &= \underbrace{(\bar{A}_2\bar{A}_1(\epsilon w_0 \sqrt{\Delta t}) + \bar{A}_2 (\epsilon w_1 \sqrt{\Delta t}) + \epsilon w_2 \sqrt{\Delta t})}_{\delta x_3^l} + e_3,\\
 e_3 &= \bar{A}_2 \bar{S}_1(\epsilon w_0\sqrt{\Delta t}) + \bar{S}_2(\bar{A}_1(\epsilon w_0\sqrt{\Delta t}) + \epsilon w_1\sqrt{\Delta t} +\\
 & \bar{S}_1(\epsilon w_0\sqrt{\Delta t})).    
 \end{align*}
 Noting that $\bar{S}_1(.)$ and $\bar{S}_2(.)$ are second and higher order terms, it follows that $e_3$ is $O(\epsilon^2)$. \\
 Suppose now that $\delta x_k = \delta x_k^l + e_k$ where $e_k$ is $O(\epsilon^2)$. Then:
 $%$\begin{align}
 \delta x_{k+1} = \bar{A}_{k}(\delta x_k^l + e_k) + \epsilon w_k \sqrt{\Delta t} + \bar{S}_{k}(\delta x_k), \\
 = \underbrace{(\bar{A}_{k} \delta x_k^l + \epsilon w_k \sqrt{\Delta t})}_{\delta x_{k+1}^l} +\underbrace{\{\bar{A}_{k}e_k + \bar{S}_{k}(\delta x_k)\}}_{e_{k+1}}.$ 
 Noting that $\bar{S}_{k}$ is $O(\epsilon^2)$ and that $e_{k}$ is $O(\epsilon^2)$ by assumption, the result follows that $e_{k+1}$ is $O(\epsilon^2)$. \\
 Now, let us take a closer look at the term $e_k$ and again proceed by induction. It is clear that $e_1 = e_1^{(2)} =0$. Next, it can be seen that $e_2 = \bar{A}_1e_1^{(2)} + \tr{\delta x_1^{l}}\bar{S}_1^{(2)}\delta x_1^l + O(\epsilon^3) = (\epsilon^2 \Delta t)  \tr{w_0} \bar{S}_1^{(2)} w_0 + O(\epsilon^3)$, which shows the recursion is valid for $k=2$ given it is so for $k=1$. \\
 Suppose that it is true for $k$. Then:
 $%$\begin{align}
 \delta x_{k+1} = \bar{A}_k \delta x_k + \bar{S}_k(\delta x_k) + \epsilon w_k \sqrt{\Delta t}
 = \bar{A}_k(\delta x_k^l + e_k) + \bar{S}_k(\delta x_k^l + e_k) + \epsilon w _k \sqrt{\Delta t} 
 = \underbrace{(\bar{A}_k \delta x_k^l + \epsilon w_k \sqrt{\Delta t})}_{\delta x_{k+1}^l}+ \underbrace{\bar{A}_k e_k^{(2)} + \tr{\delta x_k^{l}}\bar{S}_k^{(2)}\delta x_k^l}_{e_{k+1}^{(2)}}$
 $ + O(\epsilon^3),$ where the last line follows because $e_k = e_k^{(2)} + O(\epsilon^3)$, and $\bar{S}_k^{(2)}$ is the second order term of $\bar{S}_k(\cdot)$. This completes the induction and the proof.
\end{proof}

\textbf{Proof of Lemma 2.}\\
\begin{proof}
We have that:
$
\mathcal{J}^{\pi} = \sum_k \bar{c}_k + \sum_k \bar{C}_k (\delta x_k^l + e_k) + \sum_k \bar{H}_k(\delta x_k^l + e_k), 
= \sum_k \bar{c}_k + \sum_k \bar{C}_k \delta x_k^l + \sum_k \delta x_k^{l'}\bar{H}_k^{(2)} \delta x_k^l + \bar{C}_k e_k^{(2)} + O(\epsilon^3)$, where the last line of the equation above follows from an application of Lemma \ref{L1}. 
\end{proof}

\textbf{Proof of Proposition \ref{prop1}.}\\
In order to prove this result, we first need the following preparatory result. Consider the following deterministic continuous time system:
\begin{align*}
    J^{\pi}(0,x_0) &= \int_0^T \underbrace{c(x_t, \pi_t(x_t)}_{\bar{c}(t,x_t)} dt + c_T(x_T), \\
    \dot{x} &= \underbrace{f(x) + g(x)\pi_t(x)}_{\bar{f}(t,x)} + \epsilon v ,
\end{align*}
where $v(t)$ is a given continuous time input.
We rewrite the above policy evaluation equation in state-space form as follows:
$
    \dot{x} = \bar{f}(t,x) + \epsilon v, \;
    \dot{R} = \bar{c}(t,x),\;
    \dot{t} = 1,\\
    Z(t) = R(t) + c_T(x),
$
where the above equations can now be expressed in a time-invariant state space form as: $\dot{X} = F(X) + \epsilon G v$, and $Z(t) = H(X(t))$, where $X = [x, R, t]'$, $F = [\bar{f}(t,x), \bar{c}(t,x), 1]'$, $G = [I_n, 0, 0]'$ and $H (X) = R + c_T(x)$.\\
Given that the component functions $f(\cdot), ~g(\cdot), ~c(\cdot,\cdot),$ $ ~c_T(\cdot), ~\pi_t(\cdot)$ are five times continuously differentiable ($\mathcal{C}^5$) in their arguments (assumption A\ref{assump:2}), the output $Z(T)= J^{\pi}(0,x_0)$ can be expressed in terms of the inputs $v(t)$ as the unique Volterra series (Theorem 2.5 in \cite{Volterra_Krener}) where we have suppressed the dependence on $\pi$ for notational convenience:
\begin{align}
    &Z(T) = J^{(0)}(x_0) + \epsilon \int_0^T J^{(1)}(T,s) v(s) ds \nonumber\\
    & + \epsilon^2 \int_0^T \int_0^{s_1} J^{(2)}(T,s_1,s_2) v(s_1)v(s_2) ds_2ds_1 +\nonumber\\
    & \epsilon^3 \int_0^T \int_0^{s_1} \int_0^{s_2} J^{(3)}(T,s_1,s_2,s_3) v(s_3)v(s_2)v(s_1) ds_3 ds_2 ds_1 \nonumber\\
    &+ \epsilon^4 \int_0^T \int_0^{s_1} \int_0^{s_2} \int_0^{s_3} J^{(4)}(T,s_1,s_2,s_3,s_4) [v(s_4) v(s_3) \nonumber\\
    & \quad \quad v(s_2)v(s_1)] ds_3 ds_2 ds_1 + \mathcal{G}, \label{V_kernels}
\end{align}
where the Volterra kernels $J^{(k)}(.)$ are unique and continuous in 
their arguments, and $\mathcal{G}$ is an $o(\epsilon^4)$ function.

\begin{proof}
We show the result for a scalar input, the generalization to a vector input is straightforward. We first write the sample path cost in an input-output fashion in the discrete time case. Let $v(t)$ be a given input sequence, and given a discretization time $\Delta t$ such that $N = T/\Delta t$, let $v_k = v(k\Delta t)$, $k=0,1,2\cdots N-1,$ denote a piecewise constant approximation of the input. Under A\ref{assump:2}, the cost of any sample path from a given initial state $x_0$ can be expanded as follows in discrete time (where we have suppressed the explicit dependence of the different terms on $x_0$ for simplifying notation):
$%$\begin{equation}
\mathcal{V}^{\pi}_N = \mathcal{V}^{\pi,0}_N + \epsilon \mathcal{V}^{\pi,1}_N+ \epsilon^2 \mathcal{V}^{\pi,2}_N + \epsilon^3 \mathcal{V}^{\pi,3}_N + \epsilon^4 \mathcal{V}^{\pi,4}_N + \mathcal{G}_N^{\pi}, %\nonumber\\
$ where $\mathcal{V}^{\pi,0}_N$ represents the nominal/ zero input cost and
\begin{align}
&\mathcal{V}^{\pi,1}_N = \sum_{s=0}^{N-1} \mathcal{J}^{(1)}_N(N\Delta t,s\Delta t) v_s {\Delta t}, \nonumber\\
&\mathcal{V}^{\pi,2}_N = \sum_{s_1=0}^{N-1} \sum_{s_2 = 0}^{s_1} \mathcal{J}^{(2)}_N(N\Delta t, s_1 \Delta t, s_2 \Delta t) v_{s_2} v_{s_1} \Delta t^2, \nonumber\\
&\mathcal{V}^{\pi,3}_N = \sum_{s_1=0}^{N-1} \sum_{s_2 = 0}^{s_1} \sum_{s_3=0}^{s_2} \mathcal{J}^{(3)}_N(N\Delta t, s_1 \Delta t, s_2 \Delta t,s_3\Delta t)\nonumber \\
& \quad \quad \times v_{s_3}v_{s_2} v_{s_1} \Delta t^{3}, \nonumber\\
&\mathcal{V}^{\pi,4}_N = \sum_{s_1=0}^{N-1} \sum_{s_2 = 0}^{s_1} \sum_{s_3=0}^{s_2} \sum_{s_4 =0}^{s_3} \mathcal{J}^{(4)}_N(N\Delta t, s_1 \Delta t, s_2 \Delta t,s_3\Delta t,s_4\Delta t) \nonumber \\
& \quad \quad \times v_{s_4} v_{s_3}v_{s_2} v_{s_1} \Delta t^{4}, \nonumber
\end{align}
where $\mathcal{J}^{(k)}(\cdot)$ represent the piecewise constant discretized kernels corresponding to the Volterra kernels defined in \eqref{V_kernels}.
Further, the remainder function $\mathcal{G}_N^{\pi}$ is an $o(\epsilon^4)$ function. \\
Let $V^{\pi}(x_0)$ denote the cost of the trajectory under the continuous time input $v(t)$. Then it follows that $\mathcal{V}^{\pi}_N (x_0) \rightarrow V^{\pi} (x_0)$ as $N \rightarrow \infty$, regardless of the input sequence $v(t)$. Therefore, it follows that the discretized piecewise constant kernels $\mathcal{J}^{(k)}_N \rightarrow J^{(k)}$ in the $L_1$ sense as $N\rightarrow \infty$. \\
If the inputs were a discretized Wiener sequence $\omega(k\Delta t) = w_k \sqrt{\Delta t}$, where $w_k$ is a Gaussian white noise sequence, we can write the cost of a sample path as:
$\mathcal{J}^{\pi}_N = \mathcal{J}^{\pi,0}_N + \epsilon \mathcal{J}^{\pi,1}_N+ \epsilon^2 \mathcal{J}^{\pi,2}_N + \epsilon^3 \mathcal{J}^{\pi,3}_N + \epsilon^4 \mathcal{J}^{\pi,4}_N + \mathcal{R}_N^{\pi}$,
where $\mathcal{J}^{\pi,0}_N$ is the zero noise cost and
\begin{align}
&\mathcal{J}^{\pi,1}_N = \sum_{s=0}^{N-1} \mathcal{J}^{(1)}_N(N\Delta t,s\Delta t) w_s {\sqrt{\Delta t}}, \nonumber\\
&\mathcal{J}^{\pi,2}_N = \sum_{s_1=0}^{N-1} \sum_{s_2 = 0}^{s_1} \mathcal{J}^{(2)}_N(N\Delta t, s_1 \Delta t, s_2 \Delta t) w_{s_2} w_{s_1} \Delta t, \nonumber\\
&\mathcal{J}^{\pi,3}_N = \sum_{s_1=0}^{N-1} \sum_{s_2 = 0}^{s_1} \sum_{s_3=0}^{s_2}\Big( \mathcal{J}^{(3)}_N(N\Delta t, s_1 \Delta t, s_2 \Delta t,s_3\Delta t) \nonumber \\ 
& \quad \quad w_{s_3}w_{s_2} w_{s_1} (\Delta t)^{3/2}\Big), \nonumber
\end{align}
\begin{align}
&\mathcal{J}^{\pi,4}_N = \sum_{s_1=0}^{N-1} \sum_{s_2 = 0}^{s_1} \sum_{s_3=0}^{s_2} \sum_{s_4 =0}^{s_3}\Big( \mathcal{J}^{(4)}_N(N\Delta t, s_1 \Delta t, s_2 \Delta t,s_3\Delta t,\nonumber \\ 
&\quad \quad s_4\Delta t)  w_{s_4}w_{s_3}w_{s_2} w_{s_1} (\Delta t)^{2} \Big), \nonumber
\end{align}
Moreover, due to the whiteness of the noise sequence $\{w_k\}$, it follows that $E[\mathcal{J}^{\pi,1}_N] = 0$, and $E[\mathcal{J}^{\pi,3}_N] = 0$, since these terms are made of odd valued products of the noise sequences, while $E[\mathcal{J}^{\pi,2}_N], E[\mathcal{J}^{\pi,4}_N]$ are both finite owing to the finiteness of the moments of the noise values. 
Next as we take the limit of the terms above as $N \rightarrow \infty$, we obtain:
\begin{align}
    \lim_{N \rightarrow \infty} E[\mathcal{J}^{\pi,2}_N] = \int_0^T J^{(2)}(T,t,t) dt \equiv J^{\pi,1} < \infty, \nonumber\\
    \lim_{N \rightarrow \infty} E[\mathcal{J}^{\pi,4}_N] = \int_0^T \int_0^ t J^{(4)}(T,t,t,\tau,\tau) d\tau dt \equiv J^{\pi,2} < \infty, \nonumber
\end{align}
where the first equality above follows from the convergence of the discretized kernels $\mathcal{J}^{(k)}_N \rightarrow J^{(k)}$ for $k = 2,4$, while the integrals are finite owing to the continuity of the functions $J^{(2)}$ and $J^{(4)}$ as established in \eqref{V_kernels}. 
Further $\lim_{\epsilon \rightarrow 0} \epsilon^{-4} \lim_{ N \rightarrow \infty} E[\mathcal{R}_N^{\pi}] = \lim_{N \rightarrow \infty}E[\lim_{\epsilon \rightarrow 0} \epsilon^{-4} \mathcal{R}_N^{\pi}] = 0$, i.e., $\lim_{N \rightarrow \infty}E[\mathcal{R}_N^{\pi}]$ is $o(\epsilon^4)$. 
Therefore, taking expectations on both sides, we obtain:
$%$\begin{equation}
\lim_{N \rightarrow \infty} E[\mathcal{J}^{\pi}_N] = {J}^{\pi,0} + \epsilon^2 {J}^{\pi,1} + \epsilon^4 {J}^{\pi,2} + o(\epsilon^4), %\nonumber
$%\end{equation} 
where $J^{\pi,0} = \lim_{N\rightarrow \infty} \mathcal{J}^{\pi,0}_N$,
which proves the first part of the result.\\
Next, from Lemma \ref{L2}, as we take the limit $\Delta t \rightarrow 0$, it is clear that ${J}^{\pi,0}$ stems solely from the continuous-time nominal trajectory, and that ${J}^{\pi,1}$ is dependent on the continuous-time nominal and the linear closed-loop feedback. Therefore, the result follows.
\end{proof} 

\textbf{Proof of Proposition \ref{prop_e4}.}\\
\begin{proof} 
\nxx{
Using Proposition \ref{prop1}, we know that any cost function, and hence, the optimal cost-to-go function $J(t,x)$  can be expanded as:
\begin{equation} \label{f1}
J = J^0 + \epsilon^2 J^1 + \epsilon^4 J^2 + \cdots.
\end{equation}
Consider the HJB in Eq.~\eqref{DP_C} and substitute the minimizing control  $u= -\inv{R}\tr{g(x)} J^x$ (Eq.~\eqref{SOP}). This gives the PDE
\begin{align}\label{f2}
    -\frac{\partial J}{\partial t} =& \bar{l} + \frac{1}{2} \tr{(J^x)}\bar{g}\inv{R} \tr{\bar{g}} J^x + \tr{(J^x)}(\bar{f} - \bar{g}\inv{R} \tr{\bar{g}} J^x) \nonumber \\
    &+ \frac{\epsilon^2}{2}tr(J^{xx}),
\end{align}
with terminal condition $J(T,x)=c_T(x).$
 Also, $\bar{l} = l(x)$, $\bar{f} = f(x) $, $\bar{g} = g(x)$ and $tr()$ is the trace operator. Substituting Eq. \eqref{f1} into Eq. \eqref{f2} we obtain that:
\begin{align} 
&(-\frac{\partial J^0}{\partial t} - \epsilon^2 \frac{\partial J^1}{\partial t} - \epsilon^4 \frac{\partial J^2}{\partial t}+\cdots) = \bar{l} + 
\nonumber \\
&\frac{1}{2} \tr{(J^{0,x} + \epsilon^2 J^{1,x} + \cdots)} \bar{g} \inv{R} \tr{\bar{g}}(J^{0,x} + \epsilon^2 J^{1,x} + \cdots) \nonumber \\
&+\tr{(J^{0,x}+ \epsilon^2 J^{1,x}+\cdots)}\Big(\bar{f} -\bar{g} \inv{R} \tr{\bar{g}}(J^{0,x} + \nonumber\\
&\epsilon^2 J^{1,x} + \cdots)\Big) + \frac{\epsilon^2}{2} tr(J^{0,xx} + \epsilon^2 J^{1,xx}+\cdots). \label{f3}
\end{align}
Now, we equate the $\epsilon^0$, $\epsilon^2$ terms on both sides to obtain perturbation equations for the cost functions $J^0, J^1, J^2 \cdots$. \\
First, let us consider the $\epsilon^0$ term. Utilizing Eq. \eqref{f3} above, we obtain:
\begin{align}\label{f4}
-\frac{\partial J^0}{\partial t}  &= \bar{l} + \frac{1}{2} \tr{(J^{0,x})}\bar{g} \inv{R} \tr{\bar{g}}(J^{0,x}) \nonumber \\
& + \tr{(J^{0,x})} \underbrace{(\bar{f} - \bar{g} \inv{R} \tr{\bar{g}}J^{0,x})}_{\bar{f}^0},
\end{align}
with the terminal condition $J^0(T,x) = c_T(x)$.\\
Similarly, one can obtain the $J^1$ equations by equating the $O(\epsilon^2)$ terms in Eq. \eqref{f3}, which after regrouping and cancelling some of the terms yields:
\begin{align} \label{f5}
-\frac{\partial J^1}{\partial t} = \tr{(J^{1,x})} \underbrace{(\bar{f} - \bar{g}\inv{R} \tr{\bar{g}} J^{0,x})}_{= \bar{f}^0} + \frac{1}{2} tr(J^{0,xx}),
\end{align}
with terminal boundary condition $J^1(T,x) = 0$.
Note the perturbation structure of Eqs. \eqref{f4} and \eqref{f5}, $J^0(t,x)$ can be solved without knowledge of $J^1(t,x), J^2(t,x)$ etc., while $J^1(t,x)$ requires knowledge only of $J^0(t,x)$, and so on. In other words, the equations can be solved sequentially rather than simultaneously.

Now, let us consider the deterministic HJB equation in Eq.~\eqref{detDP}. Recall, $\phi(t,x)$ represents the optimal cost-to-go of the deterministic problem, and $u^d = -\inv{R}\tr{\bar{g}} \phi^x$ is the deterministic policy, analogous to the stochastic case. Substituting $u^d$ in Eq.~\eqref{detDP} gives
\begin{equation}\label{f8}
-\frac{\partial \phi}{\partial t}  = \bar{l} + \frac{1}{2}\tr{(\phi^x)} \bar{g} \inv{R} \tr{\bar{g}}\phi^x + \tr{(\phi^{x})}(\bar{f} - \bar{g}\inv{R}\tr{\bar{g}} \phi^x),
\end{equation}
with terminal condition $\phi(T,x) = c_T(x).$ 

Next, let $\varphi(t,x)$ denote the cost-to-go of the deterministic policy $u^d(\cdot)$ \textit{when applied to the stochastic system, i.e., Eq. \eqref{eq:model} with $\epsilon >0$}. Then, the cost-to-go of the deterministic policy, when applied to the stochastic system, satisfies:
\begin{align}\label{eq:varphi_pde}
    -\frac{\partial \varphi}{\partial t} =& \bar{l} + \frac{1}{2} \tr{(\phi^x)}\bar{g}\inv{R} \tr{\bar{g}} \phi^x + \tr{(\varphi^x)}(\bar{f} - \bar{g}\inv{R} \tr{\bar{g}} \phi^x) \nonumber \\
    &+ \frac{\epsilon^2}{2}tr(\varphi^{xx}),
\end{align}
with terminal condition $\varphi(T,x) = c_T(x)$. From Proposition~\ref{prop1}, we know $\varphi = \varphi^0 + \epsilon^2 \varphi^1 + \epsilon^4 \varphi^2 + \cdots$. Substituting this in Eq.~\eqref{eq:varphi_pde} gives
\begin{align}
&-\frac{\partial \varphi^0}{\partial t} - \epsilon^2\frac{\partial \varphi^1}{\partial t} - \epsilon^4 \frac{\partial\varphi^2}{\partial t}+ \cdots 
= \bar{l} + \frac{1}{2}\tr{(\phi^x)} \bar{g} \inv{R} \tr{\bar{g}}\phi^x \nonumber\\ & + \tr{(\varphi^{0,x} + \epsilon^2\varphi^{1,x} + \cdots )}\Big(\bar{f}  -\bar{g}\inv{R} \tr{\bar{g}} \phi^x\Big) \nonumber\\
& + \frac{\epsilon^2}{2} tr(\varphi^{0,xx} + \epsilon^2\varphi^{1,xx} + \cdots ). \label{f6}
\end{align}
As before, if we gather the terms for $\epsilon^0$, $\epsilon^2$, etc., on both sides of the above equation, we shall get the equations governing $\varphi^0, \varphi^1$, etc.  First, looking at the $\epsilon^0$ term in Eq. \eqref{f6}, we obtain:
\begin{equation} \label{f7}
-\frac{\partial \varphi^0}{\partial t} = \bar{l} + \frac{1}{2}\tr{(\phi^x)} \bar{g} \inv{R} \tr{\bar{g}}\phi^x + \tr{(\varphi^{0,x})}(\bar{f} - \bar{g}\inv{R}\tr{\bar{g}} \phi^x),    
\end{equation}
with the terminal condition $\varphi^0(T,x) = c_T(x)$. 

Comparing Eqs. \eqref{f7} and \eqref{f8}, it follows that $\phi(t,x) = \varphi^0(t,x)$ for all $(t,x)$. Further, comparing them to Eq. \eqref{f4}, it follows that $\varphi^0(t,x) = J^0(t,x)$, for all $(t,x)$. Also, note that the closed-loop system above, $\bar{f} - \bar{g} \inv{R} \tr{\bar{g}} \phi^x = \bar{f}^0$ (see Eq. \eqref{f4} and \eqref{f5}). 

Next, consider the $\epsilon^2$ terms in Eq. \eqref{f6}. We obtain:
\begin{align} \label{f9}
-\frac{\partial \varphi^1}{\partial t} = \tr{(\varphi^{1,x})} \underbrace{(\bar{f} - \bar{g}\inv{R}\tr{\bar{g}} \phi^x)}_{\bar{f}^0} + \frac{1}{2}tr(\varphi^{0,xx}),
\end{align}
with terminal condition $\varphi^1(T,x) = 0$. Again, comparing Eq. \eqref{f9} to Eq. \eqref{f5}, and noting that $\varphi^0 = J^0$,  it follows that $\varphi^1(t,x) = J^1(t,x)$, for all $(t,x)$. This completes the proof of the result.}  
\end{proof}
The result above has used the fact that the noise sequence $w_t$ is white. However, this is not necessary to show that $J^0(t,x) = \varphi^{0}(t,x)$ for all $(t,x)$. 
%This allows for the following general result.
\begin{comment}
\subsection{Proof of Proposition \ref{prop_gen_noise}}
\begin{proof}
Because of the dependence of the $\omega_t$ on the past, it is not necessary that it is zero mean, and thus the cost expansions will have odd terms of $\epsilon$ in the perturbation expansion leading to: $J_t(x) = J_t^0 (x) + \epsilon J_t^1 + \epsilon^2 J_t^2(x) + \epsilon^3 J_t^3 (x)+ \cdots$, and $\varphi_t(x) = \varphi_t^0 (x) + \epsilon^1 \varphi_t^1 (x) + \epsilon^3 \varphi_t^3(x) + \cdots$.\\
Substituting the above expressions into the DP equation, and equating both sides in terms of the powers of $\epsilon$ gives us the the perturbation expansion as done previously. However, albeit the equations for $J_t^k, \varphi_t^k$, $k>0$, will be different from that derived above due to the non-white nature of the noise sequence, the $J_t^0$ and $\varphi_t^0$ equations remain unchanged. Thus, using the same arguments as previously, it follows that $J_t^0(x) = \varphi_t^0(x)$ for all $(t,x)$ even in the case of non-white noise.\\
\end{proof}
\end{comment}

\textbf{Proof of Proposition \ref{OL_optimality}.}\\
\begin{proof}
    We show the scalar case, the vector case is a straightforward extension. From the existence of a smooth solution to the deterministic HJB, it follows that the Lagrange-Charpit characteristic ODEs $\dot{x} = F_q $, $\dot{q} = -F_x - qF_J$, have unique solutions in the interval $[0,T]$ from the terminal values $(x_T, q_T = c_T^x(x_T))$. Note that $q_T$ is a function of $x_T$, i.e., the terminal co-state $q_T$ is uniquely determined by the terminal state $x_T$.\\
    The family of characteristic curves originating from the terminal points $(x_T,q_T)$ on the boundary form the general solution to the HJB PDE. Since the HJB has a smooth solution (Assumption~\ref{assump:3}), the optimal cost-to-go function $J(t, x)$ is differentiable. This implies that the characteristic curves cannot intersect, i.e., there cannot be two different terminal conditions $(x_T, q_T)$ and $(x_T',q_T')$, such that the characteristic curves emanating from these points intersect at some state $x_t$ at some time $t$. To show this, suppose that such terminal conditions exist and the state/ co-state pair at time $t$ are $(x_t, q_t)$ and $(x_t,q_t')$ respectively. Note that $q_t \neq q_t'$ owing to the uniqueness of the solution of the characteristic ODEs. However, this implies that $\frac{\partial J}{\partial x}$, since it has two distinct values $q_t$ and $q_t'$, is not defined at state $x_t$ at time $t$ which violates the assumption that there exists a smooth solution to the HJB. This also implies that if a smooth solution to the HJB exists, it is unique, and is given by the characteristic ODE solutions from the terminal conditions defined above.\\
    Next, it can be seen from inspection of the Lagrange-Charpit equations \eqref{L-C-1}, \eqref{L-C-2}, that the characteristic curve that passes through the initial state $x_0$ satisfies Pontryagin's Minimum Principle (PMP) for the open-loop problem optimal control problem with initial condition $x_0$, and owing to the development above, there is a unique costate $q_0 = \frac{\partial J}{\partial x}|_{x_0}$ corresponding to the initial state $x_0$. Furthermore, since the characteristic curve passing through $x_0$ satisfies the PMP, owing to its uniqueness, it is also the global optimum solution, which establishes the fact that if the HJB admits a smooth solution, the PMP is also sufficient.
\end{proof}

\textbf{Proof of Proposition \ref{T-PFC}.}\\
\begin{proof}
Let the system model be given as
$
    \dot{x} = f(x) + g(x) u
$
where, the system matrices, its Jacobians, and Hessians are defined as in Definition~\ref{def:system_vec}.\\
Using indicial notation, the Lagrange-Charpit equations are (the subscript $t$ is ignored for the sake of simplicity):
\begin{align}
    \dot{x}_i &= f_i(x) - \Gamma_i^j \inv{R}_{jm} \Gamma_m^n q_n,\\
    \dot{q}_i &= -l_i^x - f_{ij}^x q_j + q_n \Gamma_m^n \inv{R}_{lm} \Gamma^{l,x}_{ik} q_k \label{Eq:qdot1}.
\end{align}
Performing a perturbation expansion of $\dot{x}$ around a nominal trajectory $\bar{x}$ gives 
\begin{multline}\label{Eq:deltax_dot}
    \delta \dot{x}_i = (f_{ij}^x - \Gamma_{ij}^{k,x} \inv{R}_{km} \Gamma_m^n q_n - \Gamma_k^i \inv{R}_{km} \Gamma_{nj}^{m,x} q_n) \delta x_j + \frac{1}{2}(f_{ijk}^{xx} -\\
     \inv{R}_{lm} \Gamma_m^n q_n \Gamma_{ikj}^{l,xx} - \Gamma_{ik}^{l,x} \inv{R}_{lm} \Gamma_{nj}^{m,x} q_n - \Gamma_i^l \inv{R}_{lm} \Gamma_{nkj}^{m,xx} q_n) \delta x_k \delta x_j\\ 
      - \Gamma_i^j \inv{R}_{jm} \Gamma_{m}^{n} \delta q_n  + \tilde{H}(\delta x^3) + \tilde{S}(\delta q^2). 
\end{multline}
Expanding the co-states about the nominal gives
\begin{align}
    q_i &= G_i + P_{ij} \delta x_j + H(\delta x^2) \label{Eq:q}, \\
    \delta q_i &= P_{ij}  \delta x_j + H(\delta x^2) \label{Eq:deltaq}.
\end{align}
Substituting Eq.~\eqref{Eq:q} and \eqref{Eq:deltaq}  in Eq.~\eqref{Eq:deltax_dot}, we get 
\begin{align}
    \delta \dot{x}_i &= (\bar{f}_{ij}^x - \bar{\Gamma}_{ij}^{k,x} \inv{R}_{km} \bar{\Gamma}_m^n q_n - \bar{\Gamma}_k^i \inv{R}_{km} \bar{\Gamma}_{nj}^{m,x} q_n \nonumber\\ 
    &- \bar{\Gamma}_i^l \inv{R}_{lm} \bar{\Gamma}_m^n P_{nj}) \delta x_j + H.O.T. \label{eq:deltax_dot1}
\end{align}
Let
$
    \mathcal{M}_{ij} = \bar{f}_{ij}^x - \bar{\Gamma}_{ij}^{k,x} \inv{R}_{km} \bar{\Gamma}_m^n q_n - \bar{\Gamma}_k^i \inv{R}_{km} \bar{\Gamma}_{nj}^{m,x} q_n  - \bar{\Gamma}_i^l \inv{R}_{lm} \bar{\Gamma}_m^n P_{nj}.
$
Differentiating Eq. \eqref{Eq:q} and using Eq. \eqref{eq:deltax_dot1}, we get
\begin{align}
    \dot{q}_i &= \dot{G}_i + P_{ij} \delta \dot{x}_j + \dot{P}_{ij} \delta x_j + \cdots ,\\
    \dot{q}_i &= \dot{G}_i + P_{ij}(\mathcal{M}_{jk} \delta x_k + \cdots) + \dot{P}_{ij} \delta x_j + \cdots. \label{eq:qdot cost}
\end{align}
Expanding Eq. \eqref{Eq:qdot1} upto 1st order about a nominal trajectory and substituting Eq. \eqref{Eq:q},
\begin{comment}
\begin{align}
    \dot{q}_i &= -\bar{L}_i^x - L_{ij}^{xx} \delta{x}_j - \frac{1}{2} L_{ijk}^{xxx} \delta x_j \delta x_k + \cdots \nonumber \\
    &\quad - f_{ij}^x (G_j + P_{jk} \delta x_k + \frac{1}{2} S_{jkl} \delta x_l \delta x_k + \cdots) \nonumber \\
    & \quad - \delta x_m f_{ijm}^{xx} (G_j + P_{ij} \delta x_k + \cdots) \nonumber \\
    &\quad + (G_n + P_{nk} \delta x_k + \frac{1}{2} S_{nkl} \delta x_l \delta_k + \cdots)\times \nonumber \\
    & \quad (\Gamma_m^n \inv{R}_{lm} \Gamma_{ip}^{l,x} + \Gamma_{mj}^{n,x} \delta x_j \inv{R}_{lm} \Gamma_{ip}^{l,x} \nonumber \\
    & \quad + \Gamma_m^n \inv{R}_{lm} \Gamma_{ipj}^{l,xx} \delta x_j)\times(G_p + P_{pr} \delta x_r \nonumber \\
    &\quad + \frac{1}{2} S_{prs} \delta x_r \delta x_s) 
\end{align}
\end{comment}
\begin{align}
    \dot{q}_i &= -(\bar{l}_i^x - l_{ij}^{xx} \delta{x}_j + \cdots) - \bar{f}_{ij}^x (G_j + P_{jk} \delta x_k + \cdots) \nonumber \\
    & - \delta x_m \bar{f}_{ijm}^{xx} (G_j + P_{ij} \delta x_k + \cdots)+ (G_n + P_{nk} \delta x_k + \cdots) \nonumber \\
    &\times (\bar{\Gamma}_m^n \inv{R}_{lm} \bar{\Gamma}_{ip}^{l,x}  + \bar{\Gamma}_{mj}^{n,x} \delta x_j \inv{R}_{lm} \bar{\Gamma}_{ip}^{l,x} + \bar{\Gamma}_m^n \inv{R}_{lm} \bar{\Gamma}_{ipj}^{l,xx} \delta x_j + \cdots) \nonumber \\
    &\times (G_p + P_{pr} \delta x_r + \cdots). \label{eq:qdot expansion}
\end{align}
Comparing the terms up to 1st order in $\delta x$ in Eq. \eqref{eq:qdot cost} and Eq. \eqref{eq:qdot expansion} with appropriate change in indices, we get
\begin{align}
    &\dot{G}_i = - \bar{l}_i^x - \bar{f}_{ij}^x G_j + G_n \bar{\Gamma}_m^n \inv{R}_{lm} \bar{\Gamma}_{ip}^{l,x} G_p, \label{eq:T-PFC-G index}\\ 
    &\dot{P}_{ij} = - P_{ik}(\bar{f}_{kj}^x - \bar{\Gamma}_{kj}^{l,x} \inv{R}_{lm} \bar{\Gamma}_m^n G_n) -
    (\bar{f}_{ik}^x - G_n \bar{\Gamma}_{m}^{n} \inv{R}_{lm} \bar{\Gamma}_{ik}^{l,x})P_{kj} \nonumber \\
    &- l_{ij}^{xx} - (\bar{f}_{ipj}^{xx} - G_n \bar{\Gamma}_{m}^{n} \inv{R}_{lm} \bar{\Gamma}_{ipj}^{l,xx})G_p + P_{ik} \bar{\Gamma}_{l}^{k} \inv{R}_{lm} \bar{\Gamma}_{nj}^{m,x} G_n\nonumber\\
    & + P_{ik} \bar{\Gamma}_{k}^{l} \inv{R}_{lm} \bar{\Gamma}_{m}^{n} P_{nj} + P_{nj} \bar{\Gamma}_{m}^{n} \inv{R}_{lm} \bar{\Gamma}_{ip}^{l,x} G_p \nonumber \\ 
    &+ G_n \bar{\Gamma}_{mj}^{n,x}\inv{R}_{lm}\bar{\Gamma}_{ip}^{l,x} G_p. \label{eq:T-PFC-P index}
\end{align}
Substituting $\bar{u}_l = - \inv{R}_{lm} \bar{\Gamma}_m^n G_n$ and changing indices to group terms,
Eq.~\eqref{eq:T-PFC-G index} can be written as 
$
\dot{G}_i = - \bar{l}_i^x - (\bar{f}_{ij}^x + \bar{u}_l \bar{\Gamma}_{ij}^{l,x}) G_j,
$
whose vector form is Eq.~\eqref{T-PFC-G}. Similarly, $\bar{u}_l$ can be substituted in Eq.~\eqref{eq:T-PFC-P index} and can be written as 
\begin{align}
    \dot{P}_{ij} &= - P_{ik}(\bar{f}_{kj}^x + \bar{\Gamma}_{kj}^{l,x} \bar{u}_l) -
    (\bar{f}_{ik}^x + \bar{u}_l \bar{\Gamma}_{ik}^{l,x})P_{kj} - l_{ij}^{xx} \nonumber \\
    &- (\bar{f}_{ipj}^{xx} - \bar{u}_l \bar{\Gamma}_{ipj}^{l,xx})G_p + K_{li}R_{lm}K_{mj}, \label{eq:T-PFC-P index1} \\
   \text{where, } K_{ij} &= - \inv{R}_{im}(\bar{\Gamma}_{k}^{m}P_{kj} +  \bar{\Gamma}_{kj}^{m,x}G_k). \label{eq:T-PFC-K index}
\end{align}
Eq.~\eqref{T-PFC-P} and Eq.~\eqref{T-PFC-K} are the vector form of Eq.~\eqref{eq:T-PFC-P index1} and Eq.~\eqref{eq:T-PFC-K index} respectively. 
\end{proof}

\end{appendix}

\section{Acknowledgment}
This work was supported by the NSF under grants ECCS-1637889, CDSE 1802867, and the AFOSR DDIP program under grant FA9550-17-1-0068. The simulations were conducted with the advanced computing resources provided by Texas A\&M High Performance Research Computing. 

\printbibliography 

\end{document}